\documentclass[11pt]{article}
\usepackage{xcolor}
\usepackage{graphicx}
\usepackage{subfigure}
\usepackage{snapshot}
\usepackage[margin=1in]{geometry}
\usepackage{color}
\usepackage{verbatim}
\usepackage{booktabs}
\usepackage[utf8]{inputenc}
\usepackage[numbers]{natbib}
\usepackage{enumitem}

\usepackage{bm}
\usepackage{amsfonts}
\usepackage{amssymb}
\usepackage{amsthm}
\usepackage{amsmath}
\usepackage{physics}

\usepackage{mathptmx}
\usepackage{ifpdf}
\newcommand{\mydriver}{hypertex}
\ifpdf
\renewcommand{\mydriver}{pdftex}
\fi
\usepackage[breaklinks,\mydriver]{hyperref}
\usepackage[capitalise]{cleveref}
\newtheorem{theorem}{Theorem}[section]
\newtheorem{fact}[theorem]{Fact}
\newtheorem{lemma}[theorem]{Lemma}
\newtheorem{corollary}[theorem]{Corollary}
\newtheorem{claim}[theorem]{Claim}

\newtheorem{definition}[theorem]{Definition}
\newtheorem{remark}[theorem]{Remark}

\crefname{definition}{Definition}{Definitions}
\Crefname{definition}{Definition}{Definitions}

\crefname{fact}{Fact}{Facts}
\Crefname{fact}{Fact}{Facts}

\usepackage{footnote}

\allowdisplaybreaks

\makesavenoteenv{tabular}
\makesavenoteenv{algorithm}

\usepackage{texmacs}

\usepackage{float}
\floatstyle{plain}
\newfloat{assumption}{H}{}
\floatname{assumption}{Assumption}
\crefname{assumption}{Assumption}{Assumptions}
\crefname{assumption}{Assumption}{Assumptions}

\newcommand{\eps}{\ensuremath{\varepsilon}}

\newcommand{\E}{\mathop{\mathbf{E}}}

\newcommand{\poly}{\mathrm{poly}}

\newcommand{\D}{\mathcal{D}}

\newcommand{\1}{\mathbf{1}}

\newcommand{\msqr}{\mathbin{\square}}
\newcommand{\NULL}{\bold{NULL}}
  \usepackage{nth}
  \usepackage{intcalc}

  \newcommand{\cAAAI}[1]{AAAI\ Conference\ on\ Artificial (AAAI)}

\DeclareMathAlphabet\mathbfcal{OMS}{cmsy}{b}{n}

\allowdisplaybreaks

\title{An Exponential Lower Bound for Spectral Density Estimation on Unweighted Graphs}

\author{
    Pan Peng\thanks{
    School of Computer Science and Technology, University of Science and Technology of China. \href{ppeng@ustc.edu.cn}{ppeng@ustc.edu.cn}}
    \and
	Yuyang Wang\thanks{
    School of Computer Science and Technology, University of Science and Technology of China.  \href{wangyvyang@mail.ustc.edu.cn}{wangyvyang@mail.ustc.edu.cn}}
    \and
	Joy Qiping Yang\thanks{
    The University of Sydney.  \href{yangqp5@protonmail.com}{yangqp5@protonmail.com}}
    \and
	Yichun Yang\thanks{
    Beijing Institute of Technology.  \href{yc.yang@bit.edu.cn}{yc.yang@bit.edu.cn}}
}

\date{}

\begin{document}

\maketitle

\begin{abstract}%
We study lower bounds for estimating the spectral density of the normalized adjacency matrix of a graph. Previously, Cohen-Steiner et al. [KDD 2018] proposed an algorithm for $\varepsilon$-approximate spectral density estimation in the Wasserstein-1 distance, using $2^{O(1/\varepsilon)}$ random walks initiated from uniformly random nodes in the graph. Later, Jin et al. [COLT 2023] established a nearly matching exponential lower bound for \emph{weighted} graphs, assuming the algorithm has access to samples from random walks started at random nodes. It was left open whether this lower bound could be extended to \emph{unweighted} graphs.

In this paper, we answer this question in the affirmative by proving an exponential lower bound for unweighted graphs. Specifically, we show that no algorithm can compute an $\varepsilon$-approximation to the spectrum of a normalized graph adjacency matrix with constant success probability, even when given the full transcripts of $2^{\Omega(1/\varepsilon^{1/6})}$ random walks, each of length $2^{\Omega(1/\varepsilon^{1/6})}$, started from uniformly random nodes.
\end{abstract}

\section{Introduction}
We study the problem of estimating the spectral density of the normalized adjacency matrix of an undirected graph. Let $G=(V,E,w)$ be an undirected graph on $n$ vertices with positive edge weights $w \in \mathbb{R}_{>0}^E$. Let $A_G \in \mathbb{R}^{V \times V}$ denote its adjacency matrix and let $D_G$ be the diagonal degree matrix. The normalized adjacency matrix of $G$ is defined as $N_G = D_G^{-1/2} A_G D_G^{-1/2} \in \mathbb{R}^{V \times V}$. %
Let $\lambda_1 \ge \cdots \ge \lambda_n$ be the eigenvalues of $N_G$. The \emph{spectral density} of $N_G$ (or simply of $G$), denoted by $\rho_{N_G}$, is the probability measure
\[
\rho_{N_G}(x) = \frac{1}{n} \sum_{i=1}^n \delta(x - \lambda_i),
\]
where $\delta$ denotes the Dirac delta function. The \emph{spectral density estimation (SDE)} problem asks for a probability distribution $q$ that approximates $\rho_{N_G}$ under a suitable metric. In this work, we focus on approximating $\rho_{N_G}$ to accuracy $\eps$ in the \emph{Wasserstein-1 distance} (also known as Earth Mover distance; see Definition \ref{def:wass1}). Equivalently (see e.g. \cite{Cohen-SteinerKongSohler:2018}), 
the goal is to compute (succinct representations of) estimates $\tilde{\lambda}_1 \ge \cdots \ge \tilde{\lambda}_n$ such that
\[
\frac{1}{n} \sum_{i=1}^n \lvert \tilde{\lambda}_i - \lambda_i \rvert \le \varepsilon.
\]

We refer to (the succinct representation of) such a vector ${\tilde{\lambda}_i}$ as an $\varepsilon$-approximation, in Wasserstein-1 distance (abbreviated as $W_1$-distance below), to the spectral distribution on $[-1,1]$ induced by $N_G$. More broadly, spectral density estimation is defined for any symmetric matrix and has found applications across a wide range of domains, including network science~\citep{farkas2001spectra,EikmeierGleich:2017,Cohen-SteinerKongSohler:2018,DongBensonBindel:2019,ChenTrogdonUbaru:2021,BravermanKrishnanMusco:2022}, numerical linear algebra \citep{lin2016approximating,li2019eigenvalues} and deep learning \citep{MahoneyMartin:2019}.  %

\paragraph{Sublinear-time algorithms for SDE} 
Several works have developed \emph{sublinear-time algorithms} for SDE, i.e., algorithms whose runtime is sublinear in the size of the graph's representation.

{In the \emph{random neighbor} model (also called the \emph{adaptive random walk model} in \citep{Jin23}), the algorithm is given access to two operations: (i)~sample a uniformly random vertex of $G$, and (ii)~given any specified vertex $v$, query a uniformly random neighbor of $v$. Using these operations,} Cohen-Steiner, Kong, Sohler and Valiant \citep{Cohen-SteinerKongSohler:2018} gave a randomized algorithm with query and time complexity $2^{O(1/\varepsilon)}$. Their algorithm relies solely on the transcripts of $2^{O(1/\varepsilon)}$ random walks of length $O(1/\varepsilon)$ starting from uniformly random vertices. It outputs a succinct representation of approximate eigenvalues with only $O(1/\varepsilon)$ distinct values, independent of $n$.

Moreover, in a stronger model that additionally allows the algorithm to access all neighbors of a vertex $v$ in 
time proportional to its degree\footnote{This operation can be efficiently supported in the classical \emph{adjacency list model} -- where one can query vertex degrees and the 
$i$-th neighbor. It can also be simulated with high probability in the random neighbor model at an additional $\log n$ factor, provided the degree of the queried vertex is known (see, e.g., Footnote 8 in \citep{BravermanKrishnanMusco:2022}).}, Braverman, Krishnan, and Musco \citep{BravermanKrishnanMusco:2022} gave a randomized algorithm for SDE with runtime $O(n \varepsilon^{-7})$ that outputs a length-$n$ vector representing a discrete spectral density. This was improved by Jin, Karmarkar, Musco, Sidford, and Singh \citep{jin2024faster}, who provided a randomized algorithm with $O(n \varepsilon^{-2})$ queries and $O(n \varepsilon^{-3})$ runtime, as well as a deterministic algorithm with running time $n \cdot 2^{\tilde{O}(1/\varepsilon)}$. Their approach can also be generalized to weighted graphs, though it requires a slightly stronger query model.

Given these two classes of randomized algorithms -- one with query complexity $2^{O(1/\varepsilon)}$ and the other with $O(n \cdot \mathrm{poly}(1/\varepsilon))$ -- it is natural to ask whether the exponential dependence on $\varepsilon$ in the sublinear-time algorithm of~\cite{Cohen-SteinerKongSohler:2018} can be improved, possibly to polynomial, while still avoiding any dependence on the graph size~$n$.

Answering this question turns out to be quite challenging. From a lower-bound perspective, it is difficult to construct pairs of graphs whose spectral densities are far apart, yet which are locally very similar and therefore hard to distinguish using a small number of queries. In this direction, Jin, Musco, Sidford, and Singh \citep{Jin23} proved a lower bound of $\Omega(1/\varepsilon^2)$ queries in the random neighbor model. They also studied the \emph{non-adaptive random walk model}, in which an algorithm observes transcripts of a few random walks of short length starting from uniformly random vertices, and showed that no algorithm can compute an $\varepsilon$-accurate approximation of the spectral density with constant success probability even when given $2^{\Omega(1/\varepsilon)}$ walks of length $2^{\Omega(1/\varepsilon)}$. However, this lower bound applies only to weighted graphs.

\paragraph{Our contribution} In this paper, we study lower bounds on the query complexity of SDE, with a particular focus on \emph{whether the exponential lower bound of \cite{Jin23}, previously established for weighted graphs, extends to unweighted graphs}. This question is well motivated: most real-world networks are naturally unweighted, and most sublinear-time SDE algorithms are designed for this setting (e.g., \citep{Cohen-SteinerKongSohler:2018,BravermanKrishnanMusco:2022}), so weighted-only lower bounds leave open the possibility of much faster algorithms for unweighted graphs. From a theoretical perspective, 
{unweighted graphs impose strong structural constraints --- in particular, the low-weight-edge mechanism used in prior weighted lower bounds is entirely unavailable --- making it substantially harder to construct hard instances.}
Establishing lower bounds in this regime therefore clarifies the intrinsic limitations of sublinear access models and isolates the hardness of spectral estimation itself. Indeed, \cite{Jin23} explicitly posed this as an open question. %

In this work, we provide an affirmative answer to this question. Our result shows that even in this restricted and practically relevant setting, exponential query complexity barriers persist, thereby placing fundamental limits on what sublinear-time algorithms for SDE can achieve. Specifically, we show the following:

\begin{theorem}\label{thm:main}
{There exists universal constants $\varepsilon_0 > 0$ such that the following holds. For any $\varepsilon < \varepsilon_0$,}
no algorithm that is given access to the transcripts of $m$ random walks of length $T$, each initiated at a uniformly random vertex of an \emph{unweighted} graph $G$, can approximate the spectral density of $G$ to $\varepsilon$ accuracy in Wasserstein-1 distance with probability greater than $3/4$, unless
\[
m \cdot T > 2^{c/\varepsilon^{1/6}}
\]
for some universal constant $c > 0$.
\end{theorem}

Consequently, even for unweighted graphs, no algorithm can compute an $\varepsilon$-accurate spectral density estimate with constant probability, even when given the transcripts of
$2^{\Omega(1/\varepsilon^{1/6})}$ 
random walks of length $
2^{\Omega(1/\varepsilon^{1/6})}$ 
started from uniformly random nodes. 

We remark that our lower bound applies only to the non-adaptive random walk model. It remains an open question whether an exponential lower bound can be established for the random neighbor model (also called the adaptive random walk model in \citep{Jin23}), %
in which the algorithm is allowed to start a random walk either from a uniformly random node or from any node of its choice. At present, the best known lower bound in this adaptive model %
is the aforementioned  $\Omega(1/\varepsilon^2)$.

\subsection{Technical overview}

We briefly review the lower bound for weighted graphs established by \cite{Jin23}. Their hard instances are constructed by combining cycles of different lengths. Specifically, they consider two graph families: $R_1$, consisting of $2n$ disjoint cycles of length $\ell := \Theta(1/\varepsilon)$, and $R_2$, consisting of $n$ disjoint cycles of length $2\ell$. These two graphs satisfy $W_1(R_1, R_2) = \Omega(1/\ell)$, yet they can be easily distinguished by short random walks. To overcome this issue, they first pad both $R_1$ and $R_2$ with additional isolated vertices, and then augment the resulting graphs with a low-weight complete graph on the same vertex set. Since the added edges have very small weights, this modification does not significantly affect the $W_1$-distance, while simultaneously making the two graphs hard to distinguish via random walks. 
{Intuitively, although each auxiliary edge has small weight, the \emph{total} auxiliary weight incident to any vertex forms a constant fraction of its weighted degree. Consequently, a random walk leaves the original cycle with constant probability at each step, and with high probability does not remain long enough to complete a full traversal of the cycle. At the same time, the spectral perturbation introduced by the auxiliary edges remains controlled, preserving the $W_1$-distance gap.
}

Extending this lower bound to unweighted graphs, as noted in \citep{Jin23}, is ``\emph{surprisingly tricky}''. For instance, a natural approach might be to replace the weighted complete graph with an unweighted expander; however, doing so would significantly alter the spectra of the graphs, making the analysis considerably more challenging.

\paragraph{Our approach.}
To extend the exponential lower bound to simple unweighted graphs and to bypass the difficulties above, we avoid directly augmenting disjoint cycles with unweighted expanders. Instead, we consider \emph{Cartesian products} of unweighted expander graphs and cycle graphs. Specifically, we start with a random $d$-regular graph $H$ (an expander) and construct two product graphs
\[
G_1 = H \msqr C_1 \qquad \text{and} \qquad G_2 = H \msqr C_2,
\]
where $C_1$ consists of two disjoint cycles of length~$\ell$, and $C_2$ is a single cycle of length~$2\ell$. Here $\ell=\poly(1/\varepsilon)$. As before, the $W_1$-distance between the spectral densities of $C_1$ and $C_2$ is large, on the order of $\Omega(1/\ell)$. We next show the following:
\begin{enumerate}
    \item\label{item:techoverone} The graphs $G_1$ and $G_2$ are $\varepsilon$-far in $W_1$-distance.
    \item\label{item:techovertwo} Given access to transcripts of $m$ non-adaptive random walks of length $T$ with $m \cdot T < 2^{c/\varepsilon^{1/6}}$, no algorithm can distinguish whether the transcripts were generated from $G_1$ or from $G_2$.
\end{enumerate}

To prove~Item (\ref{item:techoverone}), we exploit the fact that the spectral density of a Cartesian product satisfies a convolution property (\Cref{fact:density_convolution}). 
Combined with recent eigenvalue concentration results for random $d$-regular graphs from \citep{huang2024optimaleigenvaluerigidityrandom} (Lemma \ref{lem:eigenvalue_concentration}), this enables us to analyze the Wasserstein distance between the spectral densities of the two product graphs after convolution:
\[
\rho_{A_{G_1}} = \rho_{A_H} \ast \rho_{A_{C_1}}, \qquad 
\rho_{A_{G_2}} = \rho_{A_H} \ast \rho_{A_{C_2}},
\]
where $\rho_{A_K}$ denotes the spectral density of the adjacency matrix $A_K$ of a graph $K$, and $\ast$ denotes the convolution product (Definition \ref{def:convolution}). 
By our construction, $C_1$ places two eigenvalues at $-2$ whereas $C_2$ places only one, with the next eigenvalue of $C_2$ at $-2\cos(\pi/\ell)$ (\Cref{fact:eigenvalue_cycle}), creating a CDF discrepancy of $\frac{1}{2\ell}$ over an interval of width $\Theta(1/\ell^2)$. Since $W_1(\rho_1,\rho_2) = \int |F_1 - F_2|\,dx$ is lower bounded by restricting the integral to any sub-interval, we focus on the interval $[-2, -2\cos(\pi/\ell)]$ where only this multiplicity difference contributes. Through a careful analysis (Lemmas~\ref{lemma:W1_convolution} and~\ref{lemma:lower_bound_rho}), we show that a $\poly(1/\ell)$ gap persists after  convolution with $\rho_d$, despite the fact that $\rho_d$ vanishes at the endpoint $-2\sqrt{d-1}$ of its support. 
Consequently, the spectral densities of our constructed graphs $G_1$ and $G_2$ also exhibit a $\poly (1/\ell)$ gap under the $W_1$ distance. Finally, since $G_1$ and $G_2$ are $(d+2)$-regular graphs by our construction, we show that the gap between the spectral densities of their normalized adjacency matrix $N_{G_1}$ and $N_{G_2}$ remains: $W_1(\rho_{N_{G_1}},\rho_{N_{G_2}})=\frac{1}{d+2}W_1(\rho_{A_{G_1}},\rho_{A_{G_2}})\geq \frac{1}{10\ell^6 d^{2}}$. See \Cref{thm:w1_lower_bound} and Corollary \ref{cor:w1_lower_bound} for details. We ensure Item (\ref{item:techoverone}) by setting $2\varepsilon=\frac{1}{20\ell^6d^2}$.%

To prove Item (\ref{item:techovertwo}), we construct a coupling $\mathcal{D}$ over the distributions of random walk transcripts (Definition~\ref{def:trans}) on $G_1$ and $G_2$ such that, with high probability, the two transcripts generated by $\mathcal{D}$ are identical. The claim then follows from the classical coupling lemma. While \cite{Jin23} also employ a coupling argument to establish indistinguishability, extending their approach to our setting presents two key challenges due to the sparsity and lack of uniform local structure in our unweighted graphs. First, the base graph $H$ does not exhibit the uniform local connectivity of a complete graph. Second, after leaving a cycle, a random walk may return and continue exploring the same cycle, rather than escaping permanently as in~\citep{Jin23}. These issues make both the construction of the coupling and the analysis of coinciding transcripts substantially more delicate.

To address these challenges, we exploit the natural projection of vertices in the Cartesian product onto the base graph (Definition~\ref{def:proj}) and introduce the notion of \emph{projected walks}. Using these notions, we derive sufficient conditions under which the coupling produces identical transcripts, which ultimately allows us to complete the proof.

\paragraph{Remark} Finally, we note that there exists an algorithm that can distinguish $G_1$ from $G_2$ using $\poly(1/\varepsilon)$ queries in the \emph{adaptive} random walk model. This follows from a structural property of the hard instances: for most vertices $u$ in $G_1$ or $G_2$, there exist two edges $e^+$ and $e^-$ incident to $u$ that each belong to exactly $d$ length-$4$ cycles. These edges are precisely those lying on the copy of $C_1$ or $C_2$ containing $u$. In the adaptive random walk model, such edges can be identified using only $\poly(d)$ queries %
(by running a breadth-first search of depth 4; since a $1-o(1)$ fraction of vertices in a random $d$-regular expander are not contained in any length-4 cycles, any such cycle will be due to $e^+$ and $e^-$); completing the corresponding cycle then requires $\poly(d,\ell) = \poly(1/\varepsilon)$ queries, which suffices to distinguish whether $u$ comes from $G_1$ or $G_2$.

\subsection{Related work}
The matrix SDE problem -- computing the full eigendecomposition -- requires at least $O(n^{\omega})$ time, where $\omega < 2.373$ is the fast matrix multiplication exponent~\cite{parlett1998symmetric,banks2023pseudospectral}. For a matrix $A$, methods for spectral density estimation that run in $o(n^{\omega})$ time were first developed for applications in condensed matter physics and quantum chemistry~\cite{Skilling:1989,SilverRoder:1994,moldovan2020pybinding}. {It is shown that for any $n \times n$ symmetric matrix $A$ with spectral density $s$, the popular Stochastic Lanczos Quadrature (SLQ) method provably computes an approximate spectral density $q$ satisfying $W_1(s,q) \leq \varepsilon \|A\|_2$ using only $O(1/\varepsilon)$ matrix-vector multiplications with $A$ when $\varepsilon=\tilde{\Omega}(1/\sqrt{n})$ \citep{ChenTrogdonUbaru:2021}. 
{The same error bound and query complexity also hold by using the Chebyshev moment matching method~\citep{BravermanKrishnanMusco:2022}. Later, Musco et al.~\citep{musco2025sharper} removed the $\varepsilon = \tilde{\Omega}(1/\sqrt{n})$ assumption; see also \citep{bhattacharjee2025improved} for further improvements.}}

Our work is motivated by these advances, by the general importance of sublinear-time graph algorithms~\citep{rubinfeld2011sublinear}, and by recent progress on sublinear-time algorithms for other spectral problems, including expander testing~\citep{GoldreichRon:2011} and spectral clustering~\citep{peng2020robust,gluch2021spectral}, etc.

\section{Preliminaries}\label{sec:pre}

Let $G=(V,E)$ be a graph with vertex set $V$ and edge set $E$. Its adjacency matrix is denoted by $A_G$, its diagonal degree matrix by $D_G$, and its normalized adjacency matrix by $N_G=D_G^{-1/2}A_GD_G^{-1/2}$. We begin by defining the two types of spectral density functions that will be frequently used in our analysis.

\begin{definition}[Empirical spectral density]
    Given a graph $G$ and its normalized adjacency matrix $N_G$ with eigenvalues $\{\lambda_i\}_{1\leq i\leq n}$, the spectral density function of $N_G$ is defined by: 
    $\rho_{N_G}(x)=\frac{1}{n}\sum_{i=1}^{n}{\delta(x-\lambda_i)}$, where $\delta(\cdot)$ denotes the Dirac delta function, which can be heuristically represented as $\int_{-\infty}^{\infty} \delta(x)dx=1$ with $\delta(0)=\infty$ and $\delta(x)=0$ for all $x\neq 0$. %
    
    Similarly, the spectral density function of adjacency matrix $A_G$ is defined by: $\rho_{A_G}(x)=\frac{1}{n}\sum_{i=1}^{n}{\delta(x-\lambda'_i)}$, where $\{\lambda'_i\}_{1\leq i\leq n}$ are the eigenvalues of $A_G$.
\end{definition}

In this paper we quantify the discrepancy between two density functions using the $W_1$ distance. 

\begin{definition}\label{def:wass1}
    Given two distributions $\rho_1,\rho_2$, define $\Psi$ to be the set of all couplings $\psi(x,y)$ between $\rho_1,\rho_2$. Then the $W_1$ distance is defined as: $W_1(\rho_1,\rho_2)=\inf_{\psi\in\Psi}{\int_{\mathbb{R}}{\int_{\mathbb{R}}{|x-y|\cdot \psi(x,y)dxdy}}}.$
\end{definition}

Compared to the original definition, we use an alternative and more convenient formula for the $W_1$ distance. %

\begin{fact}[\cite{Panaretos_2019}]
    The $W_1$ distance between two density functions $\rho_1,\rho_2:\mathbb{R}\rightarrow \mathbb{R}$ can be expressed by: $W_1(\rho_1,\rho_2)=\int_{\mathbb{R}}{|F_{1}(x)-F_2(x)|dx}$, where $F_1, F_2$ are the cumulative distribution functions (CDFs) of $\rho_1,\rho_2$ respectively.
\end{fact}

We will use the following lemma, which characterizes the eigenvalues of the cycle graph.

\begin{fact}[Theorem 5.5.1 in \cite{spielman2019sagt}]\label{fact:eigenvalue_cycle}
    The eigenvalues of the adjacency matrix of length $\ell$ cycle is $2\cos(\frac{2k}{\ell}\pi)$ (with multiplicity) for $0\leq k<\ell$.%
\end{fact}

\subsection{Cartesian product of graphs}

\begin{definition}[Cartesian product]\label{def:cartesian}
The Cartesian product $G \msqr H$ of graphs $G$ and $H$ is the graph defined as follows:
\begin{itemize}
    \item the vertex set of $G \msqr H$ is $V(G) \times V(H) = \{(u,v) \mid u \in V(G),\ v \in V(H)\}$;
    \item two vertices $(u,v)$ and $(u',v')$ are adjacent in $G \msqr H$ if and only if either
    \begin{itemize}
        \item[(a)] $u = u'$ and $v$ is adjacent to $v'$ in $H$, or
        \item[(b)] $v = v'$ and $u$ is adjacent to $u'$ in $G$.
    \end{itemize}
\end{itemize}
\end{definition}

A well-known result states that the eigenvalues of the Cartesian product $G \msqr H$ are obtained by summing the eigenvalues of $G$ and $H$ respectively.

\begin{fact}[Lemma 6.3.2 in \cite{spielman2019sagt}]\label{fact:eigenvalue_product}
    Let $\{\lambda_i\}_{1\leq i\leq n}$ be the eigenvalues of adjacency matrix $A_G$, $\{\mu_i\}_{1\leq i\leq m}$ be the eigenvalues of adjacency matrix $A_H$. Then the adjacency matrix of the Cartesian product $G \msqr H$ has eigenvalues $\{\lambda_i+\mu_j\}_{1\leq i\leq n;1\leq j \leq m}$.
\end{fact}
\begin{definition}\label{def:convolution}
The convolution between two functions $f$ and $g$ is defined as: $f\ast g(z)=\int_{\mathbb{R}}{f(x)g(z-x)dx}.
$
\end{definition}

We state the following fact about the spectral density $\rho_{A_{G \msqr H}}$ of the product graph $G \msqr H$: the spectral density of the product is the convolution of the spectral densities of the factors. We will defer its proof to \Cref{proof:fact:density_convolution}.

\begin{fact}[Cartesian product and convolution]\label{fact:density_convolution}
It holds that: $
\rho_{A_{G \msqr H}} = \rho_{A_G} \ast \rho_{A_H}.
$
\end{fact}

\subsection{Random $d$-regular graphs and their spectrum}
Our hard instances are based on the random $d$-regular graphs defined as follows. 
\begin{definition}
A random $d$-regular graph on $n$ vertices is chosen uniformly from the set of all $n$-vertex $d$-regular graphs.
\end{definition}

According to classical results, the expected spectral density of the adjacency matrix for a random $d$-regular graph with sufficiently large $n$ vertices is the Kesten-McKay distribution \citep{MCKAY1981203}: 
\[\rho_d (x)  = \begin{cases}
\frac{d\sqrt{4(d-1)-x^2}}{2\pi (d^2-x^2)} &  \text{if } x\in[-2\sqrt{d-1},2\sqrt{d-1}],\\
0 & \text{otherwise}
\end{cases}\]

{Since $\rho_{A_H}$ is a discrete measure, we approximate it by the continuous Kesten--McKay density $\rho_d$, to which the analytic arguments of Lemmas~\ref{lemma:W1_convolution} and~\ref{lemma:lower_bound_rho} apply.}
We now prove that the gap between density functions $\rho_{A_H}$ and $\rho_d$ is small, where $H$ is the $n$-vertex random $d$-regular graph, using concentration results of Lemma \ref{lem:eigenvalue_concentration}. We defer its proof to \Cref{proof:lem:limit_distribution}.

\begin{lemma}\label{lem:limit_distribution}
It holds that {$W_1(\rho_{A_H},\rho_d)\leq  3\sqrt{d-1}\cdot n^{-1+o(1)}\leq n^{-0.99}$ with probability $1-n^{-\alpha}$ for constant $\alpha\geq 1$.}
\end{lemma}

\section{The hard instances}\label{sec:hardinstance}

We now define our construction of the hard instances. An illustration is provided in Figure~\ref{fig:construction}.

\begin{definition}[Hard instances]\label{def:hardinstances}
Let $d$ be some constant, and $n$ and $\ell$ be integers such that $n\geq \ell$. Fix $H = (V_H, E_H)$ to be an $n$-vertex random $d$-regular graph. Let $C_1 = (V_{C_1}, E_{C_1})$ and $C_2 = (V_{C_2}, E_{C_2})$ be graphs on $2\ell$ vertices, where $C_1$ consists of two vertex-disjoint cycles of length $\ell$, while $C_2$ consists of a single cycle of length $2\ell$. We define the hard instances:
\begin{itemize}
    \item $G_1 = H \msqr C_1$, the Cartesian product of $H$ and $C_1$;
    \item $G_2 = H \msqr C_2$, the Cartesian product of $H$ and $C_2$.
\end{itemize}
\end{definition}

\begin{figure}[h!]
    \centering
    \subfigure[$G_1$]{
    \includegraphics[scale=0.4]{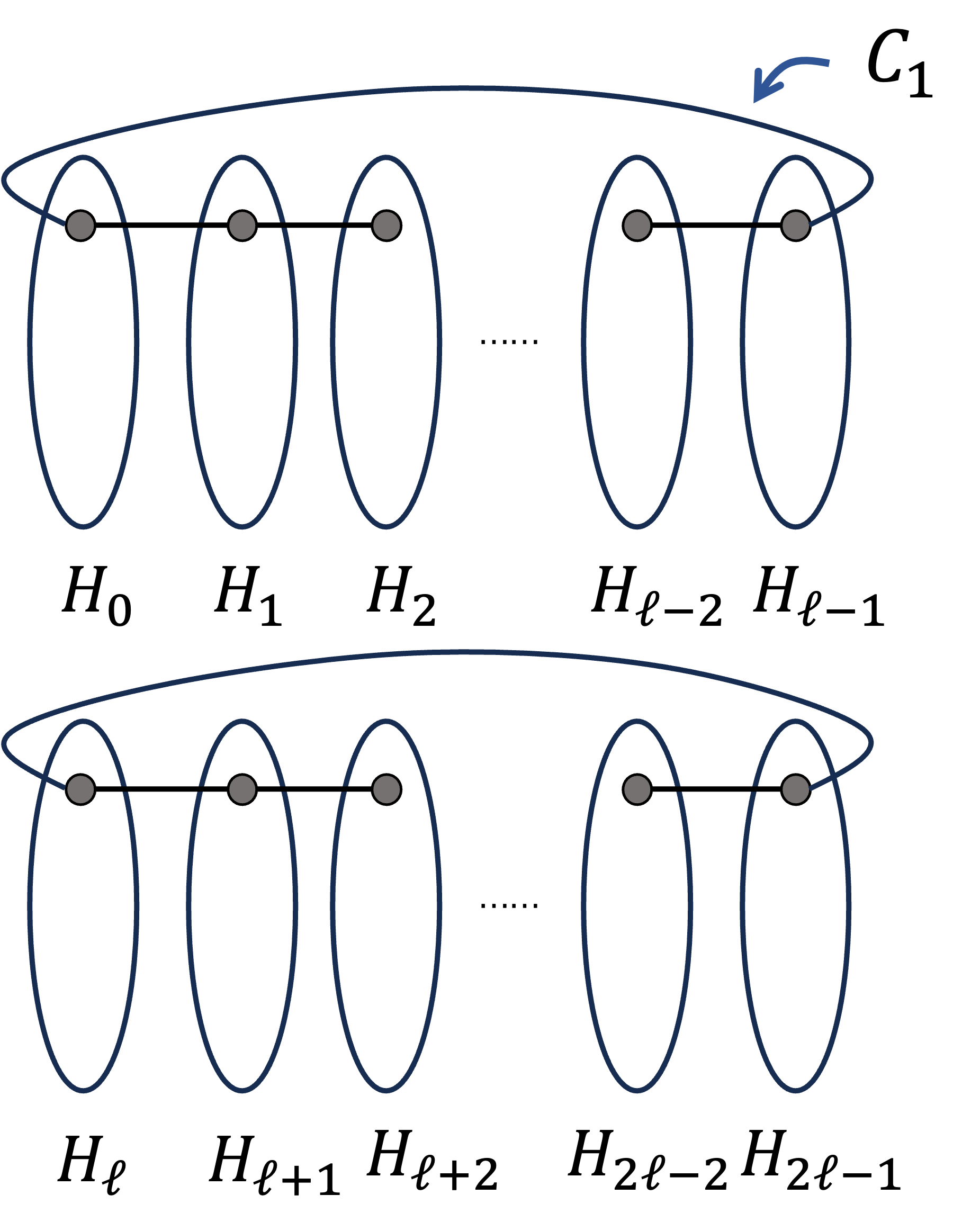} %
    }
    \subfigure[$G_2$]{
    \includegraphics[scale=0.4]{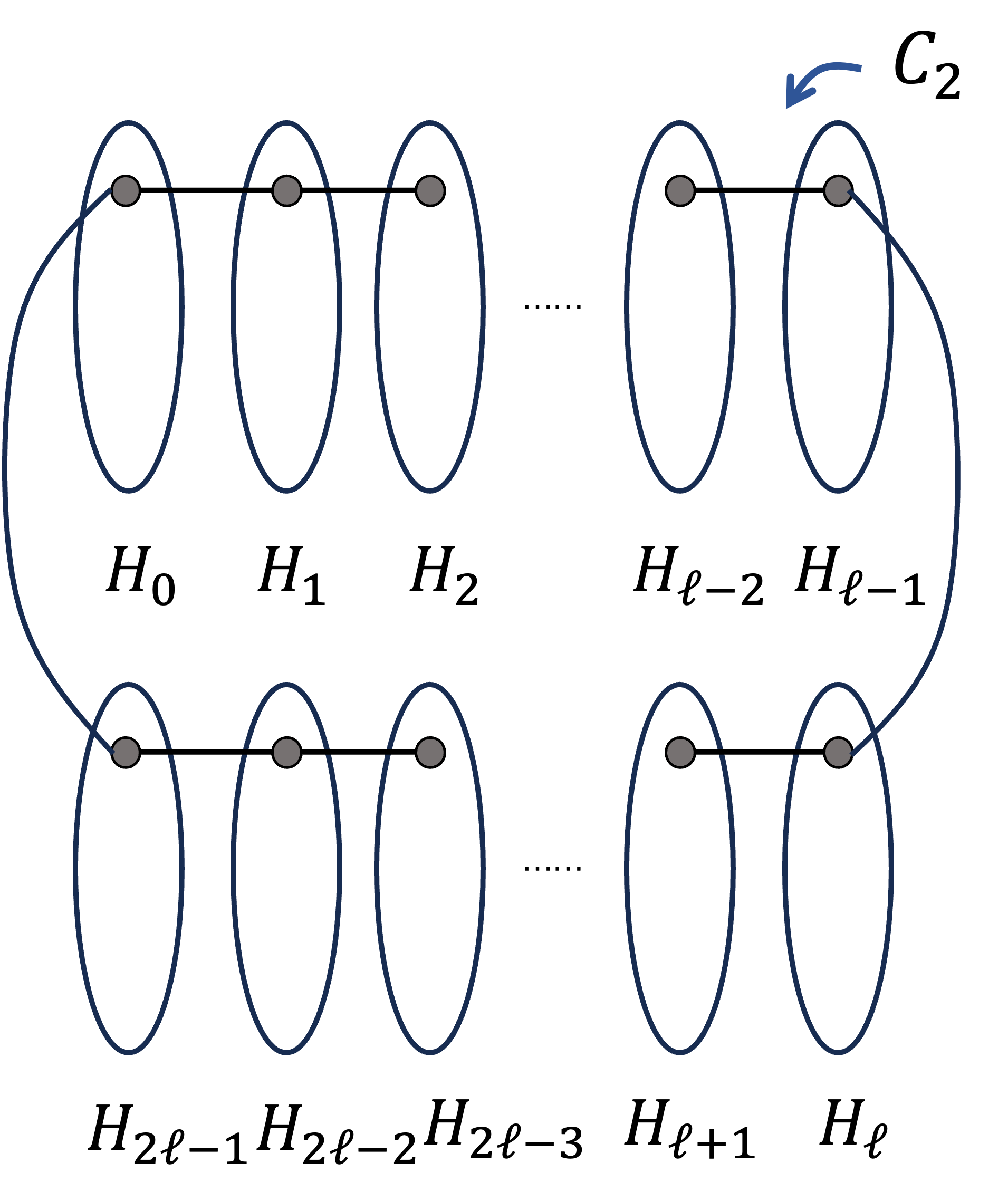} %
    }
    \caption{Illustration of our constructions $G_1$ and $G_2$.}%
    \label{fig:construction}
\end{figure}

We remark that in our final construction, we fix $d \ge 7$ as a constant, and for a sufficiently small constant $\varepsilon$, we set $\ell = \poly(1/\varepsilon)$ and $n = 2^{\Omega(\ell)}$ (see \Cref{sec:together} for details).

\section{Large $W_1$-distance between spectral densities of $G_1$ and $G_2$}\label{sec:largedistance}

In this section we prove the following theorem, showing that there is a polynomial large gap between $\rho_{A_{G_1}}$ and $\rho_{A_{G_2}}$ in terms of Wasserstein-1 distance with high probability. %

\begin{theorem}\label{thm:w1_lower_bound}
Let $G_1$ and $G_2$ be as defined in Definition \ref{def:hardinstances} %
with constant $d\geq7$ and $n\geq \ell^7d^2$. %
Then it holds that    $W_1(\rho_{A_{G_1}},\rho_{A_{G_2}})\geq \frac{\pi}{20 \ell^6 d^{3/4}} -O(n^{-0.99})$, with probability $1-n^{-\Omega(1)}$.
\end{theorem}

\begin{remark}
It is likely that the above $W_1$-distance lower bound of $\Omega(1/\ell^6)$ is not optimal. In fact, the lower bound in Theorem~\ref{thm:w1_lower_bound} may be improvable to $\Omega(1/\ell^2)$; see the simulation results in Appendix~\ref{sec:experiments} for supporting evidence. Determining the optimal lower bound of $W_1(\rho_{A_{G_1}},\rho_{A_{G_2}})$ remains an open problem.
   
\end{remark}

\noindent We have the following corollary of the $W_1$-distance of $\rho_{N_{G_1}}$ and $\rho_{N_{G_2}}$.

\begin{corollary}\label{cor:w1_lower_bound}
Let $G_1$ and $G_2$ be as defined in Definition \ref{def:hardinstances}. Suppose that %
{$d\geq7$ and $n\geq \ell^7d^2$}. 
Then it holds that  
\(W_1 (\rho_{N_{G_1}},\rho_{N_{G_2}}) \geq \frac{1}{10\ell^6d^{2}},
\)
with probability $1-n^{-\Omega(1)}$.
\end{corollary}
\begin{proof}
Note that $G_1$ and $G_2$ are $(d+2)$-regular graphs. By the relation between the eigenvalues of normalized adjacency matrix and eigenvalues of adjacency matrix, the fact that $d$ is constant and {$n\geq \ell^7d^2$}, %
and \Cref{thm:w1_lower_bound}, we have

\hfill
\(W_1 (\rho_{N_{G_1}},\rho_{N_{G_2}}) \geqslant \frac{1}{d+2}\cdot \frac{\pi}{20\ell^6d^{3/4}} -O(n^{-0.99}) \geq \frac{1}{10\ell^6d^{2}}.
\)\hfill
\end{proof}

\subsection{Proof of \Cref{thm:w1_lower_bound}}
Now we prove \Cref{thm:w1_lower_bound}. In the following, we assume that $d\geq7$ and $n\geq \ell^7d^2$. We provide the following Lemma, where we defer its proof to \Cref{proof:lem:rhoagc1c2}.

\begin{lemma}\label{lem:rhoagc1c2}
With probability at least $1-n^{-1}$, it holds that
\[
W_1(\rho_{A_{G_1}}, \rho_{A_{G_2}})
\ge
W_1 \bigl( \rho_d \ast \rho_{A_{C_1}}, \rho_d \ast \rho_{A_{C_2}} \bigr)
- O\!\left(n^{-0.99}\right).
\]
\end{lemma}

\noindent In the following, we compute the lower bound for the $W_1$ distance $W_1 (\rho_d \ast \rho_{A_{C_1}}, \rho_d \ast \rho_{A_{C_2}}) $. First, we show two lemmas (Lemma \ref{lemma:W1_convolution} and Lemma \ref{lemma:lower_bound_rho}) that when combined give a lower bound on $W_1 (\rho_d \ast \rho_{A_{C_1}}, \rho_d \ast \rho_{A_{C_2}})$. We defer their proofs to \Cref{proof:lemma:W1_convolution} and \ref{proof:lemma:lower_bound_rho}, respectively.

\begin{lemma}\label{lemma:W1_convolution}
It holds that
\[ W_1 (\rho_d \ast \rho_{A_{C_1}}, \rho_d \ast \rho_{A_{C_2}})  \geq \int_{- 2 \sqrt{d - 1} - 2}^{- 2 \sqrt{d - 1} - 2 \cos \left(
  \frac{\pi}{\ell} \right)} \frac{1}{2 \ell}  \frac{(t+2)^2+4\sqrt{d-1}(t+2)+4(d-1)}{2 \left( 2 \sqrt{d - 1}
  + t + 2 \right)}\cdot  \rho_d (t + 2) d t.
  \]
\end{lemma}

\begin{lemma}\label{lemma:lower_bound_rho}
For any $d \ge 3$, it holds that
\[
\rho_d(x) \;\ge\; \frac{\sqrt{x + 2\sqrt{d-1}}}{2\pi d^{3/4}}
\quad \text{for all } x \in \bigl(-2\sqrt{d-1},\, -1.99\sqrt{d-1}\bigr).
\]
\end{lemma}

\noindent By Lemma \ref{lemma:W1_convolution} and Lemma \ref{lemma:lower_bound_rho}, we have

\begin{align*}
 W_1 (\rho_d \ast \rho_{A_{C_1}}, \rho_d \ast \rho_{A_{C_2}}) & \geq \int_{- 2 \sqrt{d - 1} - 2}^{- 2 \sqrt{d - 1} - 2 \cos \left(
  \frac{\pi}{\ell} \right)} \frac{1}{2 \ell}  \frac{(t+2)^2+4\sqrt{d-1}(t+2)+4(d-1)}{2 \left( 2 \sqrt{d - 1}
  + t + 2 \right)}\cdot  \rho_d (t + 2) d t.\\
  & =  \int_{- 2 \sqrt{d - 1}}^{- 2 \sqrt{d - 1} + 2 - 2 \cos \left(
  \frac{\pi}{\ell} \right)} \frac{1}{2 \ell}  \frac{t^2 + 4 \sqrt{d -
  1}\cdot t+ 4 (d - 1) }{2 \left( 2 \sqrt{d - 1} + t \right)} \cdot \rho_d (t) d t\\
  & \geqslant  \int_{- 2 \sqrt{d - 1}}^{- 2 \sqrt{d - 1} + 2 - 2 \cos \left(
  \frac{\pi}{\ell} \right)} \frac{1}{2 \ell}  \frac{t^2 + 4 \sqrt{d -
  1}\cdot t+4(d-1)}{2 \left( 2 \sqrt{d - 1} + t \right)} \cdot \frac{\sqrt{t + 2 \sqrt{d
- 1}}}{2\pi d^{3/4}} d t\\
  & =  \frac{1}{2 \ell}  \int_{- 2 \sqrt{d - 1}}^{- 2 \sqrt{d - 1} + 2 - 2 \cos
  \left( \frac{\pi}{\ell} \right)} \frac{\left( t + 2 \sqrt{d - 1} \right)^2 
  }{2 \left( \sqrt{t + 2
  \sqrt{d - 1}} \right)} \cdot \frac{1}{2\pi d^{3/4}} d t\\
  & =  \frac{1}{2 \ell}  \int_0^{2 - 2 \cos \left( \frac{\pi}{\ell} \right)}
  \frac{t^2}{2\sqrt{t} } \cdot \frac{1}{2\pi d^{3/4}} d t
   =  \frac{1}{8 \ell\cdot \pi d^{3/4}}  \int_0^{2 - 2 \cos \left( \frac{\pi}{\ell} \right)}
  t^{3 / 2} d t\\
  & = \frac{1}{8 \ell\cdot \pi d^{3/4}}\cdot \frac{2}{5} \left(2-2\cos \left(\frac{\pi}{\ell}\right)\right)^{5/2} \geq 
   \frac{\pi}{20 \ell^6 d^{3/4}}  .
\end{align*}

\noindent This finishes the proof of Theorem \ref{thm:w1_lower_bound}. %

\section{Indistinguishability of $G_1$ and $G_2$}\label{sec:indisting}

\subsection{Basic tools}

Recall from Definition~\ref{def:cartesian} that in $G \msqr H$, for any fixed $u \in V(G)$, the vertex set $\{(u,v)\mid v\in V(H)\}$ induces a copy (layer) of $H$, and symmetrically for $G$. For the hard instances $G_1$ and $G_2$ (Definition~\ref{def:hardinstances}), each is the Cartesian product of a $d$-regular graph $H$ (expander) and a $2$-regular graph (cycle), and hence is $(d+2)$-regular. Every vertex lies in a unique $H$-layer and a unique cycle layer, which we refer to as the expander and the cycle of the vertex, respectively. We define the projection accordingly.

\begin{definition}[Projection onto the base graph]\label{def:proj}
Let $G_1 = H \msqr C_1$. The projection onto $H$ is the mapping
$p_H^1 : V(G_1) \to V(H)$ defined by $p_H^1(u,v)=u$ for all $(u,v)\in V(H)\times V(C_1)$.

The projection $p_H^2$ is defined analogously for $G_2 = H \msqr C_2$.
\end{definition}

\begin{remark}
Slightly abusing notation, although we \emph{relabel the vertices} using a uniformly random permutation of the integers $\{1, \dots, 2n\ell\}$, the projections $p_H^1$ and $p_H^2$ are always understood with respect to the underlying Cartesian product structure, rather than the specific vertex labels.
\end{remark}

Now we define the distribution of random walk transcripts on $G_1$ and $G_2$.

\begin{definition}[Transcript distribution]\label{def:trans}
    For $m$ non-adaptive random walks, each of length $T$, a random walk transcript $S$ is a collection of $m$ individual walks, $S = \{S_1, \ldots, S_m\}$, where each $S_i$ is a sequence of $T$ node labels $v_{i,0}, \ldots, v_{i,T}$ corresponding to the nodes visited by the walk. Let $\mathcal{D}_{G_1}$ and $\mathcal{D}_{G_2}$ denote the probability distributions over transcripts generated by performing the walks on $G_1$ and $G_2$, respectively, where the nodes are labeled using a uniformly random permutation of the integers $1, \ldots, 2n\ell$.
\end{definition}

\subsection{Indistinguishability via coupling}
This section bounds the TV distance between the transcript distributions using the coupling method, indicating indistinguishability of the hard instances under the non-adaptive random walk model.

\begin{lemma}\label{lem:agreed}
Let $\ell \geq 10 \log_d n$. For $m$ non-adaptive walks of length $T$, the TV distance between $\D_{G_1}$ and $\D_{G_2}$ is bounded by $d_{TV}(\D_{G_1},\D_{G_2}) \leq \frac{m^2T^3}{n^c}$, where $c$ is a small universal constant.
\end{lemma}

To prove this, we define a coupling $\D$ between $\D_{G_1}$ and $\D_{G_2}$. The coupling $\D$ is a joint distribution over a pair of random walk transcripts $S^1$ and $S^2$, such that the marginal distribution of $S^1$ equals $\D_{G_1}$, and the marginal distribution of $S^2$ equals $\D_{G_2}$. We show that $S^1$ and $S^2$ are identical with high probability, which establishes the closeness by the classical coupling lemma:
\[
d_{TV}(\D_{G_1},\D_{G_2}) \leq \Pr_\D[S^1 \neq S^2].
\]

\paragraph{Projected walk.}
To introduce the coupling, we describe the random walk on the constructed graph $G_1$. The $i$-th walk starts from $v_{i,1}$ chosen uniformly at random from $V(G_1)$. At each step, the current vertex $u$ has $d+2$ neighbors: $d$ in its expander and $2$ on its cycle. Hence with probability $\frac{d}{d+2}$ the walk performs an \emph{expander step}, and with probability $\frac{2}{d+2}$ a \emph{cycle step}.

We project the walk onto $H$ via $p_H^1$ (Definition~\ref{def:proj}), obtaining the \emph{projected walk}. Starting from $p_H^1(v_{i,1})$, each expander step $v_{i,k}\to v_{i,k+1}$ induces a step $p_H^1(v_{i,k}) \to p_H^1(v_{i,k+1})$. In contrast, a cycle step moves to a vertex $u'$ aligned with $u$ on $H$, implying $p_H^1(u)=p_H^1(u')$, so the projected walk stays. The walk on $G_2$ can be interpreted analogously. 
Thus, a random walk on $G_1$ (resp.\ $G_2$) corresponds, under projection $p_H^1$ (resp.\ $p_H^2$), to a lazy random walk on $H$ with laziness $\frac{2}{d+2}$. This insight contributes to both the design of coupling and the proof of Lemma~\ref{lem:agreed}.

\paragraph{Coupling.} 
We next define the coupling $\D$ by describing a random process that explicitly generates random walk transcripts $S^1$ and $S^2$. The process ensures that $S^1$ and $S^2$ are distributed according to $\D_{G_1}$ and $\D_{G_2}$. To this end, we employ a \emph{lazy labeling} procedure, which assigns labels to vertices upon their first visit during the random walks. We maintain dictionaries 
$L_1 : V_1 \to \{1,\ldots,2n\ell\}$ and 
$L_2 : V_2 \to \{1,\ldots,2n\ell\}$; initially $L_i(v)=\NULL$ for every $v\in V_i$ and $i\in\{1,2\}$. Once a label $j$ is assigned to $v$, i.e. $L_i(v)\gets j$, all subsequent queries return $j$. In describing the coupling, we also refer to the cycle containing a vertex $v$ and to its \emph{left} and \emph{right} neighbors on the cycle.

We now describe the process.
\begin{enumerate}[label=\textbf{\arabic*.}, wide, labelwidth=!, labelindent=0pt]

\item Choose a uniformly random permutation $\Pi$ of $\{1,\cdots,2n\ell\}$. Let $\Pi(j)$ denote the $j$-th label in the permutation. Initialize $j\gets 1$.

\item For $k = 1,\cdots,m$:
\begin{enumerate}[wide, labelwidth=!, labelindent=0pt]

\item Choose independent, uniformly random vertices $v^1_{k,0}\in G_1$ and $v^2_{k,0}\in G_2$.  
For $i \in \{1,2\}$, if $L_i(v^i_{k,0})=\NULL$, set $L_i(v^i_{k,0})\gets \Pi(j)$. Increment $j\gets j+1$.

\item During this iteration, maintain a variable $x$ tracking the \emph{depth} of the current vertex. Initialize $x=0$. Every cycle step increments (right) or decrements (left) $x$ by $1$.

\item Define $p^1_{k,i}:=p_H^1(v^1_{k,i})$ and $p^2_{k,i}:=p_H^2(v^2_{k,i})$.  
Initialize $H^1_k$ (resp.\ $H^2_k$) as a subgraph of $H$ containing only the isolated vertex $p^1_{k,0}$ (resp.\ $p^2_{k,0}$).

\noindent We use $H^1_k$ and $H^2_k$ to maintain the subgraphs induced by the projected walks (ignoring self-loops), updating them after each step. For every $i\in\mathbb{Z}$, let $H^1_{k,i}$ (resp.\ $H^2_{k,i}$) denote the subgraph induced by projected walk steps taken at depth $x=i$.

\noindent We aim to maintain that $H^1_k$ and $H^2_k$ remain isomorphic trees. Accordingly, we maintain a \emph{consistent isomorphic mapping} 
$f_k : V^1_k \to V^2_k$, 
initially mapping the two starting vertices. During this iteration, we maintain the invariant
\[
(\bigstar)\quad f_k(p^1_{k,i}) = p^2_{k,i} \quad \text{for all } i.
\]
If this invariant cannot be preserved, we say that the mapping $f_k$ \emph{fails}; otherwise $H^1_k$ and $H^2_k$ remain isomorphic.

\item For $i=1,\cdots,T$:
\begin{enumerate}[wide, labelwidth=!, labelindent=1pt]

\item If $f_k$ fails, choose $v^1_{k,i+1}$ and $v^2_{k,i+1}$ uniformly from the neighbors of $v^1_{k,i}$ and $v^2_{k,i}$, respectively. Otherwise, \textbf{with probability $\frac{2}{d+2}$} perform a cycle step and \textbf{with probability $\frac{d}{d+2}$} perform an expander step, defined as follows.

\begin{itemize}[leftmargin=*]

\item \textbf{(Cycle step)}
\begin{itemize}
    \item \textbf{With probability $\frac{1}{2}$:} set $v^1_{k,i+1}$ (resp.\ $v^2_{k,i+1}$) to be the left neighbor of $v^1_{k,i}$ (resp.\ $v^2_{k,i}$), and update $x \gets x-1$.
    \item \textbf{With probability $\frac{1}{2}$:} set $v^1_{k,i+1}$ (resp.\ $v^2_{k,i+1}$) to be the right neighbor of $v^1_{k,i}$ (resp.\ $v^2_{k,i}$), and update $x \gets x+1$.
\end{itemize}

\item \textbf{(Expander step)}  
Define
\[
R^1 := \{u\in V_1 : (v^1_{k,i},u)\in E_1 \text{ and } (p^1_{k,i},p_H^1(u))\in H^1_k\},
\]
that is, neighbors of $v^1_{k,i}$ whose projections have been visited. Define $R^2$ analogously for $v^2_{k,i}$.

Since $f_k$ has not failed, $f_k(p^1_{k,i})=p^2_{k,i}$ and $H^1_k,H^2_k$ are isomorphic. Let $r:=|R^1|=|R^2|$. For each $a\in R^1$, there is a unique $b\in R^2$ with $f_k(p_H^1(a))=p_H^2(b)$.

The next step is performed as follows:
\begin{itemize}
    \item \textbf{With probability $\frac{r}{d}$:} Choose $v^1_{k,i+1}$ uniformly at random from $R^1$. Then, deterministically set $v^2_{k,i+1}$ to be the unique vertex in $R^2$ satisfying $p_H^2(v^2_{k,i+1}) = f_k(p_H^1(v^1_{k,i+1}))$.
    
    \item \textbf{With probability $\frac{d-r}{d}$:} Choose $v^1_{k,i+1}$ and $v^2_{k,i+1}$ independently and uniformly at random from the remaining $(d-r)$ expander neighbors of $v^1_{k,i}$ and $v^2_{k,i}$, respectively.
    
    If $p_H^1(v^1_{k,i+1}) \notin V^1_k$ and $p_H^2(v^2_{k,i+1}) \notin V^2_k$, then update $f_k(p_H^1(v^1_{k,i+1})) := p_H^2(v^2_{k,i+1})$; otherwise, say $f_k$ fails.
\end{itemize}

\end{itemize}

\item If $L_1(v^1_{k,i+1}) = \NULL$, set $L_1(v^1_{k,i+1}) \gets \Pi(j)$. Likewise, if $L_2(v^2_{k,i+1}) = \NULL$, set $L_2(v^2_{k,i+1}) \gets \Pi(j)$. Update $f_k$, and increment $j \gets j+1$.

\end{enumerate}
\end{enumerate}

\item Return
\begin{align*}
S^1 &= \{\{L_1(v^1_{1,0}),\ldots,L_1(v^1_{1,T})\},\ldots,\{L_1(v^1_{m,0}),\ldots,L_1(v^1_{m,T})\}\},\\
S^2 &= \{\{L_2(v^2_{1,0}),\ldots,L_2(v^2_{1,T})\},\ldots,\{L_2(v^2_{m,0}),\ldots,L_2(v^2_{m,T})\}\}.
\end{align*}

\end{enumerate}

Since the next vertex is chosen uniformly at random from the neighborhood of current vertex at every step, the above procedure is a coupling over random walk transcripts. Next we argue that $S^1 = S^2$ with high probability. To achieve this, we define the following good events and show that $S^1 = S^2$ if these events occur simultaneously.

\medskip
\noindent\textbf{Event 1:} For all $k \in [m]$, $H^1_k$ and $H^2_k$ remain trees.

\noindent\textbf{Event 2:} For all $k \in [m]$ and depth $x \neq y$, $H^1_{k,x} \cap H^1_{k,y} = \emptyset$ and $H^2_{k,x} \cap H^2_{k,y} = \emptyset$ if $\ell \mid (x-y)$.

\noindent\textbf{Event 3:} For all $j \neq k$, $H^1_j \cap H^1_k = \emptyset$ and $H^2_j \cap H^2_k = \emptyset$.
\medskip

In the following, we assume that \textbf{Events 1, 2, 3} hold. 
To show that $S^1 = S^2$, it suffices to establish the following \emph{identical labeling} property: at each step, the labels of the destination vertices $L_1(v^1_{k,i+1})$ and $L_2(v^2_{k,i+1})$ are either both $\NULL$ or are both defined and equal.  We begin by stating a few auxiliary results.

\paragraph{The mapping $f_k$ does not fail.}
Observe that $f_k$ can fail only after an expander step (say, on $G_1$), and only if this step reaches a vertex $v^1_{k,i+1}\notin R^1$ whose projection $p^1_{k,i+1}$ has already been visited by the projected walk.
Since $v^1_{k,i+1}\notin R^1$, the projected edge $(p^1_{k,i},p^1_{k,i+1})$ has not been visited before. After taking this step, the edge $(p^1_{k,i},p^1_{k,i+1})$ is added to $H^1_k$. As $H^1_k$ is already a connected tree prior to adding this edge, this creates a cycle in $H^1_k$, contradicting \textbf{Event~1}.
\paragraph{Identical labeling.}
To show the identical labeling property, we use the following claim.

\begin{claim}\label{claim:kthrandomwalk}
For the $k$-th random walk, we have 
$v_{k,i}^1=v_{k,j}^1$ if and only if $v_{k,i}^{2}=v_{k,j}^{2}$.
\end{claim}

The claim, together with \textbf{Event~3} (the transcripts of different random walks are disjoint), implies the identical labeling property. It therefore suffices to prove the claim.

\begin{proof}[of Claim \ref{claim:kthrandomwalk}]
We prove one direction. Suppose $v_{k,i}^1=v_{k,j}^1$; the other direction is symmetric. Let the depths of $v_{k,i}^1,v_{k,i}^2$ be $x$, and those of $v_{k,j}^1,v_{k,j}^2$ be $y$. Since $v_{k,i}^1=v_{k,j}^1$, we have $p_{k,i}^1=p_{k,j}^1$. By the invariant $(\bigstar)\; f_k(p^1_{k,i}) = p^2_{k,i}$, we have 
\[
p_{k,i}^2=f_k(p_{k,i}^1)=f_k(p_{k,j}^1)=p_{k,j}^2 .
\]

Now $p_{k,i}^2$ is visited at depth $x$, while $p_{k,j}^2$ is visited at depth $y$, implying $H^2_{k,x}$ and $H^2_{k,y}$ intersect. By \textbf{Event~2}, this can occur only when $x=y$. Hence $v_{k,i}^{2}$ and $v_{k,j}^{2}$ lie in the same $H$-copy and have the same projection, implying $v_{k,i}^{2}=v_{k,j}^{2}$.
\end{proof}

Finally, we show that all the aforementioned good events occur simultaneously w.p. $1-\frac{m^2T^3}{n^c}$ and we defer the corresponding analysis and the remaining proof of Lemma \ref{lem:agreed} to Appendix~\ref{sec:events}.

\section{Putting things together: Proof of \Cref{thm:main}}\label{sec:together}

In this section, we complete the proof of our main theorem. We first record a
standard two-point reduction showing that transcript indistinguishability,
together with a separation in spectral density, implies a lower bound for SDE.

\begin{lemma}[Transcript indistinguishability implies SDE hardness]
\label{lem:two-point-sde}
Let \(G_1,G_2\) be two labeled graphs on the same vertex set \([N]\), and suppose
that
\[
    W_1\bigl(\rho_{N_{G_1}},\rho_{N_{G_2}}\bigr) > 2\varepsilon .
\]
Let $\mathcal{D}_{G_i}$  denote the probability distribution over transcripts generated by performing the walks on $G_i$, for \(i=1,2\).  If
\[
    d_{\mathrm{TV}}(\mathcal{D}_{G_1},\mathcal{D}_{G_2}) < \frac12,
\]
then no algorithm which receives a transcript from \(\mathcal{D}_{G_i}\) can output an
\(\varepsilon\)-approximation to \(\rho_{N_{G_i}}\) with success probability at
least \(3/4\) for both \(i=1\) and \(i=2\).

More generally, if \(d_{\mathrm{TV}}(\mathcal{D}_{G_1},\mathcal{D}_{G_2})\le \tau\), then every such
algorithm has success probability at most \(1/2+\tau/2\) on at least one of the
two inputs.
\end{lemma}

We defer the proof of the above lemma to Section \ref{append:proofoftwopoint}. We will apply Lemma \ref{lem:two-point-sde} to the two hard transcript distributions
induced by randomly labeled versions of \(G_1\) and \(G_2\). Since relabeling
does not change the spectral density, the same two-point argument applies to
these mixture transcript distributions. By Yao's minimax principle, a
distributional lower bound for these hard input distributions rules out any
randomized algorithm that succeeds on every labeled graph in the corresponding
worst-case input class.

\begin{proof}[Proof of \Cref{thm:main}] %
For the hard instances in Definition~\ref{def:hardinstances}, fix any constant \(d \ge 7\). For any sufficiently small constant \(\varepsilon > 0\), choose \(\ell\) and \(n\) such that
$2\varepsilon = \frac{1}{20 \ell^6 d^2}$
\text{and}
$n = 2^{\ell \log d / 10} \ge \ell^7 d^2$. Let \(G_1\) and \(G_2\) denote the resulting pair of graphs. Substituting \(\ell = (40 \varepsilon d^2)^{-1/6}\) and rearranging constants, we obtain
\[
n = 2^{\ell \log d / 10} \ge 2^{c_1 / \varepsilon^{1/6}},
\]
for some constant \(c_1 > 0\).

On the one hand, by Corollary~\ref{cor:w1_lower_bound}, we have $W_1(\rho_{N_{G_1}}, \rho_{N_{G_2}}) \ge \frac{1}{10 \ell^6 d^2} > 2\varepsilon$. %
Therefore, by Lemma  \ref{lem:two-point-sde}, any algorithm that outputs an
\(\varepsilon\)-approximation to the spectral density with probability at least
\(3/4\) on both inputs must be able to distinguish the corresponding transcript
distributions.

On the other hand, Lemma~\ref{lem:agreed} implies that the total variation distance between the corresponding transcript distributions satisfies $d_{\mathrm{TV}}(\mathcal{D}_{G_1}, \mathcal{D}_{G_2}) \le \frac{m^2 T^3}{n^c}$.

Consequently, by Lemma \ref{lem:two-point-sde},  given access only to the transcripts of \(m\) non-adaptive random walks of length \(T\), no algorithm can distinguish \(G_1\) from \(G_2\) with probability at least $3/4$, whenever \(m \cdot T \le 0.01\, n^{c/3}\). Since
\[
0.01\, n^{c/3} \ge 2^{c_2 / \varepsilon^{1/6}}
\]
for an appropriate constant \(c_2 < c_1 c / 3\), this completes the proof.
\end{proof}

\section*{Acknowledgments}
This work is supported in part by NSFC Grant 62272431 and Quantum Science and Technology - National Science and Technology Major Project (Grant No. 2021ZD0302901).

\bibliographystyle{alpha}
\bibliography{citations}

\appendix
\section{Deferred proofs from \Cref{sec:pre} and \Cref{sec:largedistance}}
\label{sec:deferred_proof}
\subsection{Proof of Fact \ref{fact:density_convolution}}\label{proof:fact:density_convolution}
    By definition, $\rho_{A_G}(x)=\frac{1}{n}\sum_{i=1}^{n}{\delta(x-\lambda_i)}$ and $\rho_{A_H}(x)=\frac{1}{m}\sum_{i=1}^{m}{\delta(x-\mu_i)}$. So the following holds for any $z\in \mathbb{R}$:
\begin{eqnarray*}
  \rho_{A_G}\ast \rho_{A_H}(z)&=& \int_{\mathbb{R}}{\rho_{A_G}(x) \cdot \rho_{A_H}(z-x)dx}\\
  &=&\int_{\mathbb{R}}{\left(\frac{1}{n}\sum_{i=1}^{n}{\delta(x-\lambda_i)}\right)\cdot \left(\frac{1}{m}\sum_{j=1}^{m}{\delta(z-x-\mu_j)}\right)dx}\\
  &=& \frac{1}{nm}\sum_{i=1}^{n}{\sum_{j=1}^{m}{\int_{\mathbb{R}}{\delta(x-\lambda_i})\delta(z-x-\mu_j)dx}}\\
  &=& \frac{1}{nm}\sum_{i=1}^{n}{\sum_{j=1}^{m}{\delta(z-\lambda_i-\mu_j)}}.
\end{eqnarray*}
By Fact \ref{fact:eigenvalue_product}, the last equality is exactly the spectral density function $\rho_{A_{G\msqr H}}(z)$.\hfill$\blacksquare$\par

\subsection{Proof of Lemma \ref{lem:limit_distribution}}\label{proof:lem:limit_distribution}
We will make use of the following eigenvalue concentration result. 

\begin{lemma}[Theorem 1.1 in \citep{huang2024optimaleigenvaluerigidityrandom}]\label{lem:eigenvalue_concentration}
    Fix $d$ as constant and $n$ sufficiently large, there is a constant $\alpha_d\geq 1$, depending only on $d$, such that with probability $1-n^{-\alpha_d}$, the eigenvalues $\{\lambda_i\}_{1\leq i\leq n}$ of the adjacency matrix of a random $d$-regular graph satisfy
    \[|\lambda_i-\gamma_i|\leq \sqrt{d-1}\cdot  n^{-2/3+o(1)}(\min\{i,n-i+1\})^{-1/3}\]
    for every $2\leq i\leq n$, where $\gamma_i$ satisfies $\int_{\gamma_i}^{2\sqrt{d-1}}{\rho_d(x)dx}=\frac{i-1/2}{n-1}$ for $2\leq i\leq n$.
\end{lemma}
\begin{remark}
Compared to the original statement in \citep{huang2024optimaleigenvaluerigidityrandom}, we formulate the result for the adjacency matrix rather than their normalized matrix $A/\sqrt{d-1}$.     
\end{remark}

Next we use Lemma \ref{lem:eigenvalue_concentration} to prove Lemma \ref{lem:limit_distribution}. For notational convenience, let $\rho_\lambda = \rho_{A_H} = \frac{1}{n}\sum_{i=1}^{n} \delta(x-\lambda_i)$. Define
\[
\rho_\gamma = \frac{1}{n}\Big(\delta(x-\lambda_1) + \sum_{i=2}^{n} \delta(x-\gamma_i)\Big).
\]

Let $F_\lambda$, $F_\gamma$, and $F_d$ denote the CDFs of $\rho_\lambda$, $\rho_\gamma$, and $\rho_d$, respectively.  
By definition, 
 \[F_d(\gamma_i)=\int_{-2\sqrt{d-1}}^{\gamma_i}{\rho_d(x)dx}=1-\frac{i-1/2}{n-1}.\]
    \[F_\gamma(\gamma_i)=\frac{1}{n}|\text{number of eigenvalues } \leq \gamma_i|=1-\frac{i-1}{n}.\]

Thus, $|F_d(\gamma_i) - F_\gamma(\gamma_i)| \le 1/n$ for $2\leq i\leq n$. Additionally, for $x \in (\gamma_i, \gamma_{i-1})$, we have $F_\gamma(x) = F_\gamma(\gamma_i)$ and 
$F_d(\gamma_i) \le F_d(x) \le F_d(\gamma_{i-1}) = F_d(\gamma_i) + \frac{1}{n}$.  
Thus,
\[
|F_d(x) - F_\gamma(x)| \le |F_d(\gamma_i) - F_\gamma(\gamma_i)| + \frac{1}{n} \le \frac{2}{n},
\]
which holds for all $x \in [-2\sqrt{d-1}, 2\sqrt{d-1}]$.  
By the definition of $W_1$ distance,
\[
W_1(\rho_d, \rho_\gamma) = \int_\mathbb{R} |F_d(x) - F_\gamma(x)| dx \le 4\sqrt{d-1} \cdot \frac{2}{n} + \frac{d}{n}.
\]

Next, by Lemma~\ref{lem:eigenvalue_concentration},
\[
W_1(\rho_\gamma, \rho_\lambda) = \frac{1}{n} \sum_{i=2}^n |\lambda_i - \gamma_i|
\le \frac{\sqrt{d-1}}{n} n^{-2/3+o(1)} \cdot 2 \sum_{i=2}^{\lceil n/2 \rceil} i^{-1/3}
\le 2\sqrt{d-1} \cdot n^{-1+o(1)},
\]
holding with probability $1 - n^{-\alpha}$.  
    
Combining the two bounds gives
\[
W_1(\rho_d, \rho_\lambda) \le W_1(\rho_d, \rho_\gamma) + W_1(\rho_\gamma, \rho_\lambda) \le 3\sqrt{d-1} \cdot n^{-1+o(1)},
\]
which completes the proof.\hfill$\blacksquare$\par

\subsection{Proof of Lemma \ref{lem:rhoagc1c2}}\label{proof:lem:rhoagc1c2}
By Fact \ref{fact:density_convolution}, we have 
\[ \rho_{A_{G_1}}= \rho_{A_H} \ast \rho_{A_{C_1}} ; \rho_{A_{G_2}}= \rho_{A_H} \ast \rho_{A_{C_2}}. \]
Thus, 
\[
W_1 (\rho_{A_{G_1}}, \rho_{A_{G_2}})=W_1 (\rho_{A_H} \ast \rho_{A_{C_1}}, \rho_{A_H} \ast\rho_{A_{C_2}})
\]

We first control the effect of replacing $\rho_{A_H}$ by its limit $\rho_d$.
By Lemma~\ref{lem:limit_distribution} and Fact~\ref{fact:density_convolution}, we have %
\begin{equation}\label{eq:conv_w1:ub1}
    W_1 (\rho_{A_H} \ast \rho_{A_{C_1}}, \rho_d \ast \rho_{A_{C_1}}) \leq W_1 (\rho_{A_H}, \rho_d) \leq
   n^{-0.99} ;
\end{equation}
\begin{equation}\label{eq:conv_w1:ub2}
    W_1 (\rho_{A_H} \ast \rho_{A_{C_2}}, \rho_d \ast \rho_{A_{C_2}}) \leq W_1 (\rho_{A_H}, \rho_d) \leq
   n^{-0.99}.
\end{equation}
Here, the first inequalities in \eqref{eq:conv_w1:ub1} and \eqref{eq:conv_w1:ub2} follow from
\cite[Lemma~5.2]{santambrogio2015optimal}\footnote{Note that Item~1 of \cite[Lemma~5.2]{santambrogio2015optimal} does not require the regularizing kernel assumption; only Item~2 relies on that assumption.}.

By the triangle inequality, we have that 
\begin{eqnarray*}
   W_1 (\rho_d \ast \rho_{A_{C_1}}, \rho_d \ast \rho_{A_{C_2}})
   \leq  W_1 (\rho_d \ast \rho_{A_{C_1}}, \rho_{A_H} \ast \rho_{A_{C_1}}) &+& W_1 (\rho_{A_H} \ast \rho_{A_{C_1}}, \rho_{A_H} \ast\rho_{A_{C_2}}) \\
  &+& W_1 (\rho_{A_H} \ast \rho_{A_{C_2}}, \rho_d \ast \rho_{A_{C_2}}).
\end{eqnarray*}
Rearranging, we have
\begin{align*}
  &W_1 (\rho_{A_H} \ast \rho_{A_{C_1}}, \rho_{A_H} \ast\rho_{A_{C_2}})\\
  \geq & W_1 (\rho_d \ast \rho_{A_{C_1}}, \rho_d \ast \rho_{A_{C_2}}) - W_1 (\rho_d \ast \rho_{A_{C_1}}, \rho_{A_H} \ast \rho_{A_{C_1}}) - W_1 (\rho_{A_H} \ast \rho_{A_{C_2}}, \rho_d \ast \rho_{A_{C_2}})\\
  \geq & W_1 (\rho_d \ast \rho_{A_{C_1}}, \rho_d \ast \rho_{A_{C_2}}) - O \left(n^{-0.99} \right).
\end{align*}

\hfill$\blacksquare$\par

\subsection{Proof of Lemma \ref{lemma:W1_convolution}}\label{proof:lemma:W1_convolution}
Let \(F_d(x)=\int_{-\infty}^{x}{\rho_d(s)ds}\) be the CDF of $\rho_d$. In the following, we let 

\[
\{\lambda_i\}_{0\leq i\leq 2\ell-1}=\left\{2\cdot \cos\left(\frac{2\pi i}{\ell}\right)\right\}_{0\leq i\leq 2\ell-1}, \{r_i\}_{0\leq i\leq 2\ell-1}=\left\{2\cdot \cos\left(\frac{\pi i}{\ell}\right)\right\}_{0\leq i\leq 2\ell-1}
\]
be the eigenvalues of $A_{C_1}$ and of $A_{C_2}$, respectively. By the definition of convolution, for any $x\in \mathbb{R}$,
    \[
    \rho_d \ast \rho_{A_{C_1}}(x)=\frac{1}{2\ell}\sum_{i=0}^{2\ell-1}{\rho_d(x-\lambda_i)}
    \]
\[
   \rho_d \ast \rho_{A_{C_2}}(x)=\frac{1}{2\ell}\sum_{i=0}^{2\ell-1}{\rho_d(x-r_i)}
    \]
Let $F_1,F_2$ be the CDFs of $\rho_d \ast \rho_{A_{C_1}}$ and $\rho_d \ast \rho_{A_{C_2}}$ respectively. By definition, $F_1(x)=\frac{1}{2\ell}\sum_{i=0}^{2\ell-1}{F_d(x-\lambda_i)}$ and $F_2(x)=\frac{1}{2\ell}\sum_{i=0}^{2\ell-1}{F_d(x-r_i)}$ for any $x\in \mathbb{R}$. Therefore, 
\begin{eqnarray}
   W_1 (\rho_d \ast \rho_{A_{C_1}}, \rho_d \ast \rho_{A_{C_2}}) &=& \int_{\mathbb{R}}{\left|F_1(t)-F_2(t)\right| dt} \nonumber \\
    &= &\int_{\mathbb{R}} \frac{1}{2
   \ell}  \left| \sum_{i = 0}^{2 \ell-1} F_d (t - \lambda_i) - F_d (t - r_i) \right| d
   t. \label{eqn:w1_Fd}
\end{eqnarray}
 
Next, we use the fact that $F_d(x)=0$ for all $x \le -2\sqrt{d-1}$.  
Consider
\[
t \le -2\sqrt{d-1} - 2\cos\!\left(\frac{\pi}{\ell}\right).
\]
For any $r_i$, this implies
\[
t - r_i \le -2\sqrt{d-1} - 2\cos\!\left(\frac{\pi}{\ell}\right) - r_i.
\]
In particular, if $r_i \ge -2\cos\!\left(\frac{\pi}{\ell}\right)$, then
$t-r_i \le -2\sqrt{d-1}$, and hence $F_d(t-r_i)=0$.
An identical argument shows that $F_d(t-\lambda_i)=0$ whenever
$\lambda_i \ge -2\cos\!\left(\frac{\pi}{\ell}\right)$.

Without loss of generality, assume that $\ell$ is even.
Using the identities
\[
\lambda_i = 2\cos\!\left(\frac{2\pi i}{\ell}\right),
\qquad
r_i = 2\cos\!\left(\frac{\pi i}{\ell}\right),
\]
the condition $\lambda_i, r_i < -2\cos\!\left(\frac{\pi}{\ell}\right)$
occurs only for
\[
\lambda_{\ell/2}=\lambda_{3\ell/2}=r_\ell
=2\cos(\pi)=-2.
\]
{(If $\ell$ is odd, the condition $\lambda_i,r_i<
-2\cos\!\left(\frac{\pi}{\ell}\right)$ occurs only for $r_\ell=-2$;
the argument below remains unchanged.)}

Therefore, we have
\begin{eqnarray}
  && \int_{\mathbb{R}} \frac{1}{2 \ell}  \left| \sum_{i = 0}^{2 \ell-1} F_d (t -
  \lambda_i) - F_d (t - r_i) \right| d t \nonumber\\
  & \geq & \int_{- 2 \sqrt{d - 1}
  - 2}^{- 2 \sqrt{d - 1} - 2 \cos \left( \frac{\pi}{\ell} \right)} \frac{1}{2 \ell} 
  \left| \sum_{i = 0}^{2 \ell-1} F_d (t - \lambda_i) - F_d (t - r_i) \right| d t \nonumber\\
  &=& \int_{- 2 \sqrt{d - 1}
  - 2}^{- 2 \sqrt{d - 1} - 2 \cos \left( \frac{\pi}{\ell} \right)} \frac{1}{2 \ell} 
  \left| F_d(t-\lambda_{\ell/2})+F_d(t-\lambda_{3\ell/2})-F_d(t-r_l) \right| d t \nonumber \\
  & = & \int_{- 2 \sqrt{d - 1} - 2}^{- 2 \sqrt{d - 1} - 2 \cos \left(
  \frac{\pi}{\ell} \right)} \frac{1}{2 \ell} F_d (t + 2) d t \nonumber \\
  & = & \int_{- 2 \sqrt{d - 1} - 2}^{- 2 \sqrt{d - 1} - 2 \cos \left(
  \frac{\pi}{\ell} \right)} \frac{1}{2 \ell}  \int_{- 2 \sqrt{d - 1}}^{t + 2} \rho_d
  (x) d x d t.\label{eqn:w1_intermediate}
\end{eqnarray}

We now use Lemma~\ref{thm:concave_rho_d}, which states that for $d\ge 7$ the density
$\rho_d(t)$ is concave. Let $\alpha\in(0,1)$. By Jensen’s inequality,
\begin{align*}
\rho_d\!\left((1-\alpha)(-2\sqrt{d-1})+\alpha(t+2)\right)
&\ge (1-\alpha)\rho_d(-2\sqrt{d-1})+\alpha\rho_d(t+2) \\
&= \alpha\rho_d(t+2),
\end{align*}
where the equality uses the fact that $\rho_d(x)=0$ for all $x\le -2\sqrt{d-1}$.

Let
\[
s=(1-\alpha)(-2\sqrt{d-1})+\alpha(t+2).
\]
Solving for $\alpha$ yields
\[
\alpha=\frac{s+2\sqrt{d-1}}{2\sqrt{d-1}+t+2}.
\]
Substituting this expression back gives
\[
\rho_d(s)\;\ge\;
\frac{s+2\sqrt{d-1}}{2\sqrt{d-1}+t+2}\cdot\rho_d(t+2).
\]

Thus, we have
\begin{eqnarray}
  \int_{- 2 \sqrt{d - 1}}^{t + 2} \rho_d (s) d s & \geq & \int_{- 2
  \sqrt{d - 1}}^{t + 2} \frac{s + 2 \sqrt{d - 1}}{2 \sqrt{d - 1} + t + 2}
  \rho_d (t + 2) d s \nonumber \\
  & = & \frac{(t + 2)^2 - 4 (d - 1) }{2 \left( 2 \sqrt{d - 1}
  + t + 2 \right)} \cdot \rho_d (t + 2) +\frac{2\sqrt{d-1}\cdot (t+2+2\sqrt{d-1})}{2 \sqrt{d - 1} + t + 2}\cdot \rho_d (t + 2) \nonumber\\
  & =& \frac{(t+2)^2+4\sqrt{d-1}(t+2)+4(d-1)}{2 \left( 2 \sqrt{d - 1}
  + t + 2 \right)}\cdot  \rho_d (t + 2).\label{eqn:w1final}
\end{eqnarray}

Therefore, by Eq. (\ref{eqn:w1_Fd}), (\ref{eqn:w1_intermediate}) and (\ref{eqn:w1final}), we have
\[ W_1 (\rho_d \ast \rho_{A_{C_1}}, \rho_d \ast \rho_{A_{C_2}})  \geq \int_{- 2 \sqrt{d - 1} - 2}^{- 2 \sqrt{d - 1} - 2 \cos \left(
  \frac{\pi}{\ell} \right)} \frac{1}{2 \ell}  \frac{(t+2)^2+4\sqrt{d-1}(t+2)+4(d-1)}{2 \left( 2 \sqrt{d - 1}
  + t + 2 \right)}\cdot  \rho_d (t + 2) d t.
  \]
\hfill$\blacksquare$\par

\subsection{Proof of Lemma \ref{lemma:lower_bound_rho}}\label{proof:lemma:lower_bound_rho}
Notice that 
\[
\rho_d(x) = \frac{d \sqrt{4(d-1)-x^2}}{2\pi (d^2-x^2)} \;\ge\; \frac{1}{2\pi d} \sqrt{4(d-1)-x^2}
\quad \text{for } x \in \bigl(-2\sqrt{d-1},\, -1.99\sqrt{d-1}\bigr).
\]

Set $\lambda = x + 2\sqrt{d-1}$, so that $\lambda \in (0, 0.01 \sqrt{d-1})$. Then
\begin{align*}
\rho_d(x) 
&\ge \frac{1}{2\pi d} \sqrt{4(d-1) - x^2} 
= \frac{1}{2\pi d} \sqrt{4(d-1) - (\lambda - 2\sqrt{d-1})^2} \\
&= \frac{1}{2\pi d} \sqrt{4\sqrt{d-1}\, \lambda - \lambda^2} 
\ge \frac{1}{2\pi d} \sqrt{3.99 \sqrt{d-1}\, \lambda} 
\ge \frac{\sqrt{\lambda}}{2\pi d^{3/4}}.
\end{align*}

This completes the proof.\hfill$\blacksquare$\par

\section{Lower-bounding the probabilities of good events in Section~\ref{sec:indisting}}\label{sec:events}
We restate the good events defined in Section~\ref{sec:indisting} for convenience, to complete the proof of Lemma~\ref{lem:agreed}.

\medskip
\noindent\textbf{Event 1:} For all $k \in [m]$, $H^1_k$ and $H^2_k$ remain trees.

\noindent\textbf{Event 2:} For all $k \in [m]$ and depth $x \neq y$, $H^1_{k,x} \cap H^1_{k,y} = \emptyset$ and $H^2_{k,x} \cap H^2_{k,y} = \emptyset$ if $\ell \mid (x-y)$.

\noindent\textbf{Event 3:} For all $j \neq k$, $H^1_j \cap H^1_k = \emptyset$ and $H^2_j \cap H^2_k = \emptyset$.
\medskip

To this end, we apply the following lemma for random $d$-regular graph.

\begin{lemma}[Proposition 4.1 in \cite{bauerschmidt2019local}]\label{lem:local tree}
    Let $\delta > 0$ and $\omega \geq 1$ be an integer. There is $0 <c < \delta/(2\omega +2)$ such that, if $L \leq c\log_{d-1} n$, then the following holds for a uniformly chosen random $d$-regular $G$ on vertex set $[n]$ for sufficiently large $n$, with probability at least $1 - o(n^{-\omega+\delta})$.
    \begin{itemize}
        \item Most $L$-neighborhoods are trees:
        \[
        |\{v \in [n] \mid \text{the subgraph $N_L(v,G)$ contains a cycle}\}| \leq n^\delta.
        \]
    \end{itemize}
    Here, $N_L(v,G)$ is the subgraph of $G$ induced by $L$-hop neighborhood rooted at $v$.
\end{lemma}

\paragraph{Bound on Event 1.}
Consider the $i$-th projected walk on $H$, say $W_i = (v_0,v_1,\cdots,v_T)$.
\begin{claim}
    $\Pr[\text{$W_i$ contains a cycle}] \leq \frac{T^2}{n^c}$, where $c > 0$ is a small universal constant.
\end{claim}
\begin{proof}
    We define $S := \{v \mid \text{$N_L(v,H)$ is not a $d$-ray tree} \}$, for $L = 0.01 \cdot \log n/\log d$. Lemma \ref{lem:local tree} tells that $|S| \leq n^{0.1}$ w.p. $\geq 1 - o(\frac{1}{n})$ over randomness of $H$, by setting $\delta = 0.1$ and $\omega=2$. We conditioned on $H$ s.t. $|S| \leq n^{0.1}$.

    Define event $\mathcal{E}:=\{\text{$W_i$ hits at least one vertex in $S$}\}$. Since $\Pr[v_i = u] = \frac{1}{n}$ for every $i$ and $u \in V_H$, a direct argument using Markov's inequality tells that 
    \[
    \Pr[\mathcal{E}] \leq T \cdot n^{-0.9},
    \]
    hence we have
    \begin{align*}
        &\quad\Pr[\text{$W_i$ contains a cycle}] \\
        &\leq \Pr[\text{$W_i$ contains a cycle of length $\leq L$}] + \Pr[\text{$W_i$ only contains cycle of length $\geq L+1$}] \\
        &\leq \Pr[\mathcal{E}] + \Pr[\text{$W_i$ only contains cycle of length $\geq L+1$}] \\
        &\leq T \cdot n^{-0.9} + \Pr[\text{$W_i$ only contains cycle of length $\geq L+1$}].
    \end{align*}

    To bound the latter term, we bound the probability that the $(t+1)$-th step finishes a cycle. Here we only care for $t \in [L, T-1]$, since any cycle is of length $\geq L+1$.
    
    We define $C_t = \{v_0,v_1,\dots,v_{t-L}\}$ to be vertices visited during the first $t-L$ steps of $W_i$. We denote by $F_{t+1}$ the event that the $(t+1)$-th step of $W_i$ finishes a cycle of length $\geq L+1$, which occurs only when $v_{t+1} \in C_t$. Let $\1_{F_{t+1}}$ be the indicator of $F_{t+1}$, then for any fixed $C_t$, we have
    \begin{align*}
        \E[\1_{F_{t+1}} \mid C_t] &\leq \sum_{u \in C_t} \Pr[v_{t+1} = u \mid C_t]\\
        &\leq \sum_{u \in C_t} e_{v_{t-L}}\cdot P^{L+1} \cdot e_{u} \\
        & \leq \sum_{u \in C_t} (\frac{1}{n}+\lambda^{L+1}) \\
        & \leq T \cdot n^{-c}
    \end{align*}
    The last inequality follows from the fact that $H$ is a (near) Ramanujan graph (i.e. $\lambda = \max\{\lambda_2,|\lambda_n|\} = \Theta(1/\sqrt{d})$) and $L = \Theta(\log_d n)$. Furthermore, we have
    \begin{align*}
        \E[\1_{F_{t+1}}] &= \E_{C_t}[\E[\1_{F_{t+1}} \mid C_t]] \leq T \cdot n^{-c}.
    \end{align*}
    Hence
    \begin{align*}
        &\quad\Pr[\text{$W_i$ only contains cycle of length $\geq L+1$}] \\
        &\leq \sum_{t = L}^{T-1} \E[F_{t+1}] \\
        &\leq T^2 \cdot n^{-c}.
    \end{align*}
    This finishes the proof.
\end{proof}

    By a union bound over all $m$ projected walks, we have
    \[
    \Pr[\text{Event 1}] \geq 1 - \frac{mT^2}{n^c}.
    \]

\paragraph{Bound on Event 2.}
We consider the $T$-step lazy random walk on $H$. We denote by $v_i$ the vertex visited at $i$-th step, and $x_i$ the depth. Here, the notion \emph{depth} is defined analogously in the coupling process. We define event $\mathcal{E}_k$ characterizing the following bad event:
\[
\text{$\mathcal{E}_k$: $\exists s < t$ s.t. $x_t-x_s = k\ell$ and $v_s = v_t$.}
\]
In other words, the event that we visit the same node after traversing a cycle of length $\ell$ for the $k$-th time. Note that $1 \leq |k| \leq \lfloor\frac{T}{\ell}\rfloor$.

\begin{claim}
    $\Pr[\mathcal{E}_k] \leq O(\frac{T^2}{n})$ if $\ell \geq 10 \log_d n$.
\end{claim}
\begin{proof}
    \begin{align*}
        \Pr[\mathcal{E}_k] &\leq \sum_{\substack{s<t \\ s,t \in [T]}} \Pr[v_s=v_t, x_t-x_s=k\ell] \\
        &\leq \sum_{u\in V_H} \sum_{\substack{s<t \\ s,t \in [T]}} \Pr[v_s=v_t=u, x_t-x_s=k\ell] \\
        &= \sum_{u\in V_H} \sum_{\substack{s<t \\ s,t \in [T]}} \Pr[v_t=u, x_t-x_s=k\ell \mid v_s=u] \cdot \Pr[v_s = u] \\
        &= \frac{1}{n}\sum_{u\in V_H}\sum_{\substack{s<t \\ s,t \in [T]}} \Pr[v_t=u, x_t-x_s=k\ell \mid v_s=u] \\
        &\leq \frac{1}{n}\sum_{u\in V_H}\sum_{\substack{t-s \geq |k|\ell \\ s,t \in [T]}} \Pr[v_t=u \mid v_s=u] \\
        &= \frac{1}{n}\sum_{u\in V_H}\sum_{\substack{t-s \geq |k|\ell \\ s,t \in [T]}} e_u^TP^{t-s}e_u \\
    \end{align*}
    The fourth equality is due to the fact that we are in the non-adaptive random walk setting and thus the node we visit at the $s$ step is uniformly random.
    As discussed above, $e_u^TP^{t-s}e_u \leq \frac{1}{n}+\lambda^{t-s}$, hence
    \begin{align*}
        \Pr[\mathcal{E}_k] &\leq \frac{1}{n}\sum_{u\in V_H}\sum_{\substack{t-s \geq |k|\ell \\ s,t \in [T]}} e_u^TP^{t-s}e_u \\
        &\leq \frac{1}{n}\sum_{u\in V_H}\sum_{\substack{t-s \geq |k|\ell \\ s,t \in [T]}} (\frac{1}{n}+\lambda^{t-s})\\
        &\leq \sum_{\substack{t-s \geq |k|\ell \\ s,t \in [T]}} (\frac{1}{n}+\lambda^{t-s})\\
        &\leq \frac{T^2}{n}+T\cdot\sum_{\substack{T \geq t^\prime \geq |k|\ell}}\lambda^{t^\prime}\\
        &\leq \frac{T^2}{n} + O(T\lambda^{|k|\ell}) \\
        &\leq \frac{2T^2}{n}.
    \end{align*}
    The last equation follows from $\lambda \approx 1/\sqrt{d}$, $\ell \geq 10\log_d n$ and $k \geq 1$.
\end{proof}
A union bound over all $k$ and $m$ independent walks gives us the following
\[
\Pr[\text{Event 2}] \geq 1 - \frac{2mT^3}{n}.
\]

\paragraph{Bound on Event 3.}
Let $S$ be the set of vertices on $H$ visited by first $k$ projected walks $W_1, W_2, \cdots, W_k$. We consider the probability that the $(k+1)$-th projected walk $W_{k+1}$ visits $S$.
\begin{claim}
    $\Pr[\text{$W_{k+1}$ visits at least one vertex in $S$}] \leq \frac{T}{n}\cdot|S|$.
\end{claim}
\begin{proof}
    Since $W_{k+1}$ starts from a stationary distribution, we have $\Pr[v_i = u] = \frac{1}{n}$ for every $i$ and $u \in V_H$, where $v_i$ is the vertex visited at the $i$-th step of $W_{k+1}$. Hence, the expected number of visits to $S$ is $\frac{|S|}{n}|T|$. A straightforward argument using Markov's inequality finishes the proof.
\end{proof}

Since $|S| \leq mT$, we have 
\[
\Pr[\text{Event 3}] \geq 1 - m \cdot \frac{T}{n} \cdot mT = 1 - \frac{m^2T^2}{n}, 
\]
by a union bound over all $m$ projected walks.

\paragraph{Putting together.}
\begin{proof}[\textbf{of Lemma \ref{lem:agreed}}]
    By a union bound over Event $1$, $2$ and $3$, we see that the output of coupling process remains same with probability at least
    \[
    1 - \frac{mT^2}{n^c} - \frac{2mT^3}{n} - \frac{m^2T^2}{n} \geq 1 - \frac{m^2T^3}{n^c}.
    \]
    This finishes the proof.
\end{proof}

\section{Deferred proofs from Section \ref{sec:together}}

\subsection{Proof of Lemma \ref{lem:two-point-sde}}\label{append:proofoftwopoint}
Suppose, for contradiction, that an algorithm \(A\) succeeds with probability at
least \(3/4\) on both \(G_1\) and \(G_2\).  We construct a distinguisher \(D\)
between \(\mathcal{D}_{G_1}\) and \(\mathcal{D}_{G_2}\).  Given a transcript \(\mathcal S\), run \(A\) on
\(\mathcal S\), obtaining \(\widehat\rho\).  Output \(1\) if
\[
    W_1(\widehat\rho,\rho_{N_{G_1}})
    <
    W_1(\widehat\rho,\rho_{N_{G_2}}),
\]
and output \(2\) otherwise.

If \(\mathcal S\sim \mathcal{D}_{G_1}\) and
\(W_1(\widehat\rho,\rho_{N_{G_1}})\le \varepsilon\), then by the triangle
inequality and the assumption that 
\(W_1(\rho_{N_{G_1}},\rho_{N_{G_2}})>2\varepsilon\), we have
\[
    W_1(\widehat\rho,\rho_{N_{G_2}})
    >
    \varepsilon
    \ge
    W_1(\widehat\rho,\rho_{N_{G_1}}),
\]
so \(D\) outputs \(1\).  Thus \(D\) is correct with probability at least \(3/4\)
when the input comes from \(\mathcal{D}_{G_1}\).  The same argument shows that \(D\) is correct
with probability at least \(3/4\) when the input comes from \(\mathcal{D}_{G_2}\).  Hence \(D\)
distinguishes \(\mathcal{D}_{G_1}\) and \(\mathcal{D}_{G_2}\) with advantage at least \(1/4\), which implies
\(d_{\mathrm{TV}}(\mathcal{D}_{G_1},\mathcal{D}_{G_2})\ge 1/2\), a contradiction.

The quantitative version follows from the standard fact that the maximum
distinguishing success probability between two distributions \(\mathcal{D}_{G_1},\mathcal{D}_{G_2}\) with
equal prior is \(1/2+d_{\mathrm{TV}}(\mathcal{D}_{G_1},\mathcal{D}_{G_2})/2\). %
\hfill$\blacksquare$\par

\section{Numerical experiments for measuring $W_1(\rho_{A_{G_1}},\rho_{A_{G_2}})$}\label{sec:experiments}

Given parameters $n$, $d$, and $\ell$, the eigenvalues of $G_1$ and $G_2$ can be computed exactly from our construction. In this section, we study the $W_1$ distance between $\rho_{A_{G_1}}$ and $\rho_{A_{G_2}}$ by varying $\ell$. Specifically, we fix $n=10{,}000$ and $d=7$, and let $\ell$ range from $10$ to $300$. Figures~\ref{fig:illustration} (a) and (b) show the eigenvalues of $G_1$ and $G_2$ for $\ell=10$, while Figure~\ref{fig:illustration} (c) plots $W_1(\rho_{A_{G_1}},\rho_{A_{G_2}})$ against $1/\ell^2$ on a logarithmic scale as $\ell$ varies.

The results indicate that $W_1(\rho_{A_{G_1}},\rho_{A_{G_2}})$ is experimentally between $1/\ell^2$ and $1/\ell^3$, suggesting that the theoretical lower bound could ideally be improved. Consequently, the exponential lower bound for spectral density estimation could potentially be strengthened using our construction. However, due to current analytical limitations, we can only rigorously establish the weaker bound stated in Theorem~\ref{thm:w1_lower_bound}, namely 
\[
W_1(\rho_{A_{G_1}},\rho_{A_{G_2}}) = \Omega(1/\ell^6).
\]

\begin{figure}
    \centering
    \subfigure[Eigenvalues of $A_{G_1}$ when $\ell=10$]{
    \includegraphics[scale=0.27]{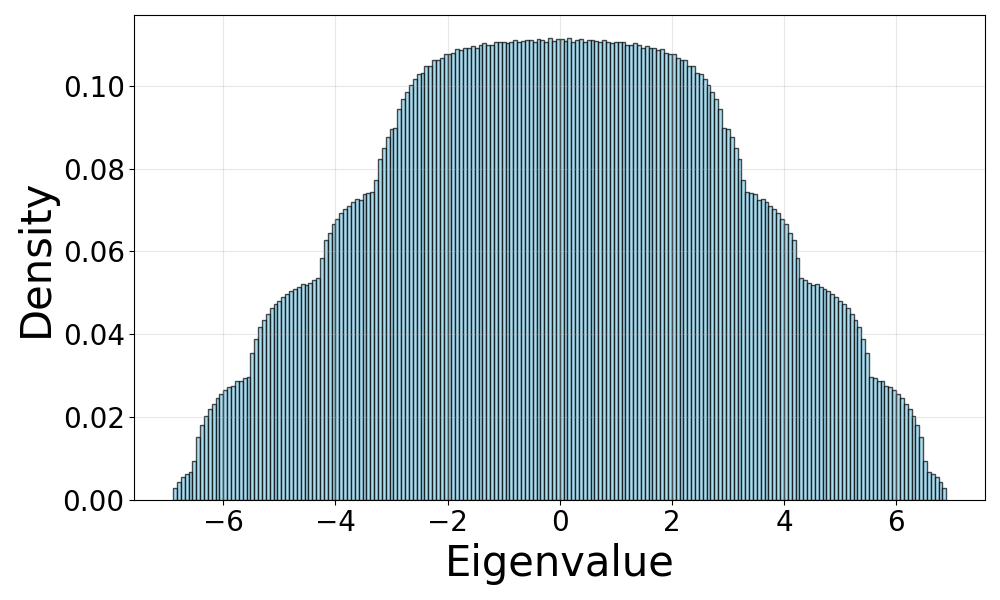} %
    }
    \hspace{0.2cm} \subfigure[Eigenvalues of $A_{G_2}$ when $\ell=10$]{
    \includegraphics[scale=0.27]{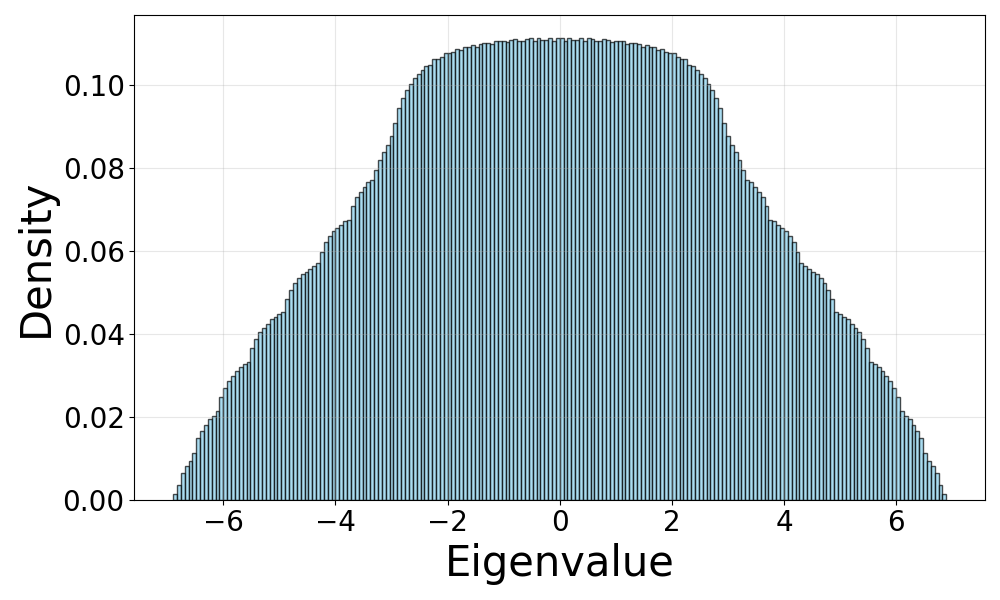} %
    }
    \hspace{0.2cm} \subfigure[$W_1(\rho_{A_{G_1}},\rho_{A_{G_2}})$ vs. $\poly \frac{1}{\ell}$ in $\log$ scales]{
    \includegraphics[scale=0.3]{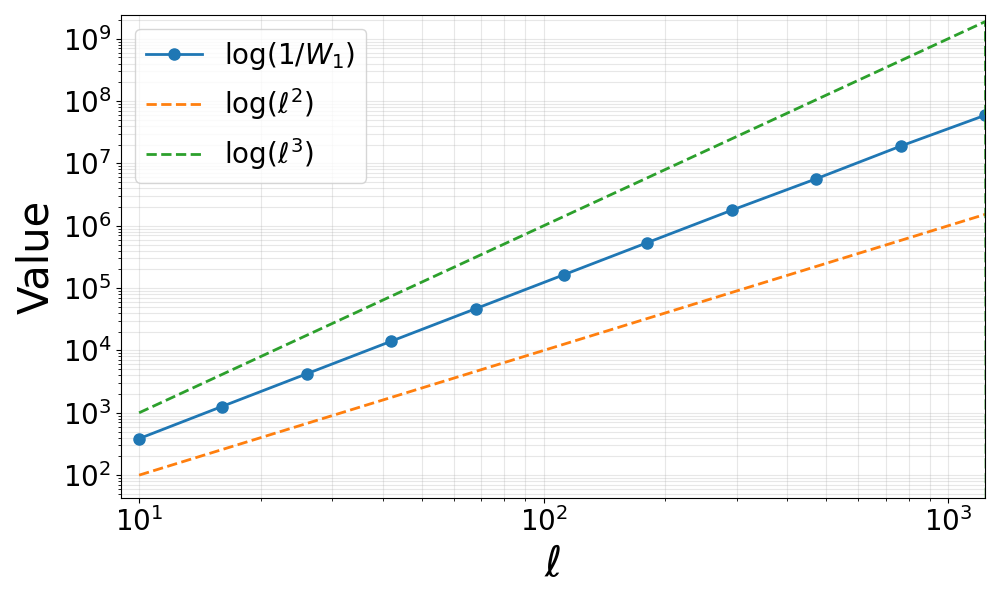} %
    }
    \caption{Testing $W_1(\rho_{A_{G_1}},\rho_{A_{G_2}})$}\vspace{-0.5cm}
    \label{fig:illustration}
\end{figure}

\section{Concavity of $\rho_d(x)$}

\begin{lemma}\label{thm:concave_rho_d}
Let $d \ge 2$ be an integer. The Kesten-McKay density function
\[
\rho_d(x) = \frac{d\sqrt{4(d-1)-x^2}}{2\pi(d^2-x^2)}, \quad |x| < 2\sqrt{d-1}
\]
is strictly concave on its support if and only if $d \ge 7$.
\end{lemma}

\begin{proof}
Let $a^2 = 4(d-1)$, so the support of $\rho_d(x)$ is $|x| < a$. Define
\[
g(x) = \sqrt{a^2 - x^2}, \quad h(x) = d^2 - x^2.
\]
Then
\[
\rho_d(x) = C \cdot \frac{g(x)}{h(x)}, \quad \text{where } C = \frac{d}{2\pi} > 0.
\]
Since $C$ is a positive constant, the concavity of $\rho_d$ is equivalent to the concavity of
\[
\varphi(x) = \frac{g(x)}{h(x)}.
\]

We compute:
\[
g'(x) = -\frac{x}{\sqrt{a^2-x^2}}, \quad h'(x) = -2x.
\]
By the quotient rule:
\begin{align*}
\varphi'(x) &= \frac{g'(x)h(x) - g(x)h'(x)}{h(x)^2}\\
&= \frac{-\frac{x}{\sqrt{a^2-x^2}}(d^2-x^2) - \sqrt{a^2-x^2}(-2x)}{(d^2-x^2)^2}\\
&= \frac{-x(d^2-x^2) + 2x(a^2-x^2)}{(d^2-x^2)^2\sqrt{a^2-x^2}}.
\end{align*}
Simplify the numerator:
\[
-x(d^2-x^2) + 2x(a^2-x^2) = x\left[-d^2 + x^2 + 2a^2 - 2x^2\right] = x\left[2a^2 - d^2 - x^2\right].
\]
Define constant
\[
B = 2a^2 - d^2 = 8(d-1) - d^2.
\]
Then
\[
\varphi'(x) = \frac{x\left[B - x^2\right]}{(d^2-x^2)^2\sqrt{a^2-x^2}}.
\]

Next step we compute the second derivative. Let
\[
N(x) = x(B - x^2), \quad D(x) = (d^2-x^2)^2\sqrt{a^2-x^2}.
\]
Then $\varphi'(x) = N(x)/D(x)$. By the quotient rule:
\[
\varphi''(x) = \frac{N'(x)D(x) - N(x)D'(x)}{D(x)^2}.
\]
Since $D(x) > 0$ on $|x| < a$, the sign of $\varphi''(x)$ is determined by the numerator:
\[
M(x) = N'(x)D(x) - N(x)D'(x).
\]
We compute:
\[
N'(x) = B - 3x^2.
\]
To compute $D'(x)$, take logarithms:
\[
\ln D(x) = 2\ln(d^2-x^2) + \frac{1}{2}\ln(a^2-x^2),
\]
\[
\frac{D'(x)}{D(x)} = 2\cdot\frac{-2x}{d^2-x^2} + \frac{1}{2}\cdot\frac{-2x}{a^2-x^2} = -\frac{4x}{d^2-x^2} - \frac{x}{a^2-x^2}.
\]
Thus
\[
D'(x) = -D(x)x\left[\frac{4}{d^2-x^2} + \frac{1}{a^2-x^2}\right].
\]
Substitute into $M(x)$:
\begin{align*}
M(x) &= (B-3x^2)D(x) - x(B-x^2)\left(-D(x)x\left[\frac{4}{d^2-x^2} + \frac{1}{a^2-x^2}\right]\right)\\
&= D(x)\left[B-3x^2 + x^2(B-x^2)\left(\frac{4}{d^2-x^2} + \frac{1}{a^2-x^2}\right)\right].
\end{align*}
Therefore,
\[
\varphi''(x) = \frac{M(x)}{D(x)^2} = \frac{B-3x^2 + x^2(B-x^2)\left(\frac{4}{d^2-x^2} + \frac{1}{a^2-x^2}\right)}{D(x)}.
\]
Since $D(x) > 0$, we have $\varphi''(x) < 0$ if and only if
\[
M(x) := B-3x^2 + x^2(B-x^2)\left(\frac{4}{d^2-x^2} + \frac{1}{a^2-x^2}\right) < 0.
\]

For $d \ge 7$, we show that $M(x) < 0$ for all $|x| < a$. First, compute $B$:
\[
B = -d^2 + 8d - 8.
\]
For $d = 7$, $B =  -1$. For $d > 7$, $B < -1$ since $B$ is decreasing for $d > 4$. Now consider the three terms in $M(x)$:
\begin{align*}
T_1 &= B - 3x^2, \\
T_2 &= \frac{4x^2(B-x^2)}{d^2-x^2}, \\
T_3 &= \frac{x^2(B-x^2)}{a^2-x^2}.
\end{align*}
Since $B \le -1 < 0$ and $x^2 \ge 0$, we have $T_2 \le 0$ and $T_3 \le 0$. Thus,
\[
M(x) = T_1 + T_2 + T_3 \le T_1 = B - 3x^2<0 \quad \text{for all } |x|< a.
\]

Therefore, for $d \ge 7$, $\varphi''(x) < 0$ for all $ |x|< a$. Hence $\rho_d(x)$ is strictly concave. Finally we show that when $d \le 6$, $\rho_d(x)$ is not concave. Actually for $d \le 6$, we have $B > 0$. In particular, at $x = 0$:
\[
M(0) = B > 0,
\]
so $\varphi''(0) > 0$, meaning $\rho_d$ is convex near $x = 0$. Therefore, $\rho_d$ cannot be concave on the entire interval.
\end{proof}

\end{document}